\documentclass[epjST]{svjour}
\usepackage{bbm,amsmath,graphicx,amssymb}
\begin{document}
\def\id{\mathbbm I}
\def\set#1{{\cal #1}}\def\sH{\set{H}}
\def\sa{\fontfamily{roman}\selectfont}
\def\Tr{\hbox{\rm Tr}}
\newcommand{\media}[1]{\langle #1 \rangle}
\def\sT{{\scriptscriptstyle T}}
\def\kket#1{|#1\rangle\!\rangle}
\def\bbra#1{\langle\!\langle #1|}
\newcommand{\smz}{{\scriptscriptstyle Z}}
\newcommand{\smx}{{\scriptscriptstyle X}}
\newcommand{\smy}{{\scriptscriptstyle Y}}
\newcommand{\sms}{{\scriptscriptstyle S}}
\newcommand{\smp}{{\scriptscriptstyle P}}
\newcommand{\smt}{{\scriptscriptstyle T}}
\newcommand{\smn}{{\scriptscriptstyle N}}
\newcommand{\smb}{{\scriptscriptstyle B}}
\newcommand{\smk}{{\scriptscriptstyle K}}
\newcommand{\sma}{{\scriptscriptstyle A}}
\newcommand{\smc}{{\scriptscriptstyle C}}
\newcommand{\smab}{{\scriptscriptstyle \!AB}}
\newcommand{\smac}{{\scriptscriptstyle \!AC}}
\newcommand{\cE}{{\cal E}}
\newcommand{\sme}{{\scriptscriptstyle E}}
\newcommand{\smq}{{\scriptscriptstyle Q}}
\spnewtheorem{postulate}{Postulate}{\bf}{\rm}
\title{The modern tools of quantum mechanics}
\subtitle{A tutorial on quantum states, measurements,
and operations}
\author{Matteo G A Paris\inst{1,2,3}}
\institute{Dipartimento di Fisica dell'Universit\`a degli Studi di
Milano, I-20133 Milano, Italia, EU 
\and 
CNISM - Udr Milano, I-20133 Milano, Italia, EU.
\and
\email{matteo.paris@fisica.unimi.it}}
\abstract{We address the basic postulates of quantum mechanics and point out
that they are formulated for a closed isolated system. Since we are
mostly dealing with systems that interact or have interacted with the
rest of the universe one may wonder whether a suitable modification is
needed, or in order. This is indeed the case and this tutorial is
devoted to review the modern tools of quantum mechanics, which are
suitable to describe states, measurements, and operations of realistic,
not isolated, systems.  We underline the central role of the Born rule
and and illustrate how the notion of density operator naturally emerges,
together with the concept of purification of a mixed state. In
reexamining the postulates of standard quantum measurement theory, we
investigate how they may be formally generalized, going beyond the
description in terms of selfadjoint operators and projective
measurements, and how this leads to the introduction of generalized
measurements, probability operator-valued measures (POVMs) and detection
operators. We then state and prove the Naimark theorem, which elucidates
the connections between generalized and standard measurements  and
illustrates how a generalized measurement may be physically implemented.
The "impossibility" of a joint measurement of two non commuting
observables is revisited and its canonical implementation as a
generalized measurement is described in some details.  The notion of
generalized measurement is also used to point out the heuristic nature
of the so-called Heisenberg principle.  Finally, we address the basic
properties, usually captured by the request of unitarity, that a map
transforming quantum states into quantum states should satisfy to be
physically admissible, and introduce the notion of complete positivity
(CP). We then state and prove the Stinespring/Kraus-Choi-Sudarshan
dilation theorem and elucidate the connections between the CP-maps
description of quantum operations, together with their operator-sum
representation, and the customary unitary description of quantum
evolution. We also address transposition as an example of positive map
which is not completely positive, and provide some examples of
generalized measurements and quantum operations.
}
\maketitle
\newpage
\tableofcontents
\section{Introduction}
Quantum information science is a novel discipline which addresses how
quantum systems may be exploited to improve the processing,
transmission, and storage of information.  This field has fostered new
experiments and novel views on the conceptual foundations of quantum
mechanics, and also inspired much current research on coherent quantum
phenomena, with quantum optical systems playing a prominent role. Yet,
the development of quantum information had so far little impact on the
way that quantum mechanics is taught, both at graduate and undergraduate
levels.  This tutorial is devoted to review the mathematical tools of
quantum mechanics and to present a modern reformulation of the basic
postulates which is suitable to describe quantum systems in interaction
with their environment, and with any kind of measuring and processing
devices.
\par
We use Dirac braket notation throughout the tutorial  and by {\em system} 
we refer to a single given degree 
of freedom (spin, position, angular momentum,...) of a physical entity. 
Strictly speaking we are going to deal with systems 
described by finite-dimensional Hilbert spaces and with observable quantities
having a discrete spectrum. Some of the results may be generalized to the
infinite-dimensional case and to the continuous spectrum.
\par
The postulates of quantum mechanics
are a list of prescriptions to summarize 
\begin{itemize}
\item[1.] how we describe the states of a physical system;
\item[2.] how we describe the measurements performed on a physical
system;
\item[3.] how we describe the evolution of a physical system, either
because of the dynamics or due to a measurement. 
\end{itemize}
In this section we
present a pico\-review of the ba\-sic po\-stu\-lates of quantum mechanics in
or\-der to intro\-duce notation and point out both i) the implicit
assumptions contained in the standard formulation, and ii) the need of 
a reformulation in terms of more general
mathematical objects. For our purposes 
the postulates of quantum mechanics may be grouped and summarized
as follows
\begin{postulate}[States of a quantum system]\label{pqs}
The possible states of a physical system correspond to normalized vectors 
$|\psi\rangle$,  $\langle \psi|\psi\rangle=1$,  of a Hilbert space 
$H$. Composite systems, either made by more than one physical object 
or by the different degrees of freedom of the same entity, are described 
by tensor product $H_1 \otimes H_2 \otimes ...$ of the corresponding 
Hilbert spaces, and the overall state of the system is a vector in the 
global space. As far as the Hilbert space description of physical
systems is adopted,  then we have the {\em superposition principle}, 
which says that if $|\psi_1\rangle$ and $|\psi_2\rangle$ 
are possible states of a system, then also any (normalized) linear 
combination $\alpha|\psi_1\rangle+\beta|\psi_2\rangle$, $\alpha,\beta\in
{\mathbbm C}$, $|\alpha|^2+|\beta|^2=1$ of the two states is an admissible 
state of the system.
\end{postulate}
\begin{postulate}[Quantum measurements]\label{pqm}
Observable quantities are described by Hermitian operators $X$. Any 
hermitian operator $X=X^\dag$, admits a spectral decomposition 
$X=\sum_x x P_x$, in terms of its real eigenvalues $x$, which are
the possible value of the observable, and of the projectors 
$P_x=|x\rangle\langle x|$, $P_x,P_{x'}=\delta_{xx'}P_x$ on its
eigenvectors $X|x\rangle=x|x\rangle$, which form a basis for the
Hilbert space, i.e. a complete set of orthonormal states with the 
properties $\langle x|x'\rangle = \delta_{xx'}$ (orthonormality), 
and $\sum_x |x\rangle\langle x| =\id$ (completeness, we omitted to
indicate the dimension of the Hilbert space).  
The probability of obtaining the outcome $x$ from the measurement of 
the observable  $X$ is given by $p_x =\left|\langle\psi|
x\rangle\right|^2$, i.e 
\begin{align}
p_x 
= \langle\psi| P_x| \psi \rangle
=\sum_n \langle \psi| \varphi_n\rangle\langle \varphi_n |P_x|\psi\rangle
=\sum_n 
\langle \varphi_n |P_x|\psi\rangle
\langle \psi| \varphi_n\rangle
=
\hbox{Tr}\left[|\psi\rangle\langle\psi|\, P_x\right]\,,
\end{align}
and the overall expectation value by
$$\langle X\rangle = \langle\psi| X|\psi\rangle =
\hbox{Tr}\left[|\psi\rangle\langle\psi|\, X \right]\,.$$
This is the {\em Born rule}, which represents the fundamental recipe to connect the
mathematical description of a quantum state to the prediction of quantum 
theory about the results of an actual experiment. 
The state of the system {\em after} 
the measurement is the (normalized) projection of the state {\em before} 
the measurement on  the eigenspace of the observed eigenvalue, i.e.
$$
|\psi_x\rangle = \frac{1}{\sqrt{p_x}}\, P_x |\psi\rangle\,.
$$
\end{postulate}
\begin{postulate}[Dynamics of a quantum system]\label{pqd}
The dynamical evolution of a physical system is described by 
unitary operators: if $|\psi_0\rangle$ is the state of the system 
at time $t_0$ then the state of the system at time $t$ is given by 
$|\psi_t\rangle = U(t,t_0) |\psi_0\rangle$, with 
$U(t,t_0)U^\dag(t,t_0)=U^\dag(t,t_0) U(t,t_0)=\id$. 
\end{postulate}
We will denote by $L(H)$ the linear space of (linear) operators from $H$ to $H$,
which itself is a Hilbert space with scalar product provided by the
trace operation, i.e. upon denoting by $|A\rangle\rangle$ operators seen
as elements of $L(H)$, we have $\langle\langle A| B\rangle\rangle =
\hbox{Tr}[A^\dag B]$ (see Appendix \ref{apTR} for details on the trace
operation). 
\par
As it is apparent from their formulation, the postulates of quantum 
mechanics, as reported above, are about a closed isolated system. 
On the other hand, we are mostly dealing 
with system that interacts or have interacted with the rest of the 
universe, either during their dynamical evolution, or when subjected to
a measurement. As a consequence, one may wonder whether a suitable
modification is needed, or in order. This is indeed the case and 
the rest of his tutorial is devoted to review the tools 
of quantum mechanics and to present a modern reformulation of 
the basic postulates which
is suitable to describe, design and control quantum systems 
in interaction with their environment, and with any kind of measuring 
and processing devices.
\section{Quantum states}
\subsection{Density operator and partial trace}
Suppose to have a quantum system whose preparation is not completely
un\-der con\-trol. What we know is that the system is prepared in the
state $|\psi_k\rangle$ with probability $p_k$, i.e. that the system
is described by the statistical ensemble $\{p_k, |\psi_k\rangle\}$,
$\sum_k p_k=1$, where the states $\{|\psi_k\rangle\}$ are not, in general, 
orthogonal. The expected value of an observable $X$ may be evaluated 
as follows
\begin{align}\notag
\langle X\rangle & = \sum_k p_k \langle X \rangle_k = \sum_k p_k
\langle\psi_k | X| \psi_k \rangle = \sum_{n\,p\,k} p_k
\langle\psi_k|\varphi_n\rangle\langle\varphi_n |
X|\varphi_p\rangle\langle\varphi_p| \psi_k \rangle 
\\ \notag &=
\sum_{n\,p\,k} p_k
\langle\varphi_p| \psi_k \rangle\langle\psi_k|\varphi_n\rangle
\langle\varphi_n |X|\varphi_p\rangle
= \sum_{n\,p} 
\langle\varphi_p| \varrho| \varphi_n\rangle
\langle\varphi_n |X|\varphi_p\rangle
\\ \notag 
&=\sum_{p} 
\langle\varphi_p| \varrho\,X|\varphi_p\rangle
= \hbox{Tr}\left[\varrho\,X\right]\,
\end{align}
where $$\varrho = \sum_k p_k\,|\psi_k\rangle\langle\psi_k |$$
is the {\em statistical (density) operator} describing the system under investigation.
The $|\varphi_n\rangle$'s in the above formula are a basis for the
Hilbert space, and we used the trick of suitably inserting two resolutions of the
identity $\id = \sum_n |\varphi_n\rangle\langle\varphi_n|$.   
The formula is of course trivial if the $|\psi_k\rangle$'s are themselves 
a basis or a subset of a basis.  
\begin{theorem}[Density operator]
An operator $\varrho$ is the density operator 
associated to an ensemble $\{p_k, |\psi_k\rangle\}$ is and only if it 
is a positive $\varrho\geq 0$ (hence selfadjoint) operator with unit 
trace $\Tr\left[\varrho\right]=1$. 
\end{theorem}
\begin{proof}:
If $\varrho=\sum_k p_k |\psi_k\rangle\langle\psi_k |$
is a density operator then $\hbox{Tr}[\varrho]=\sum_k p_k=1$ and
for any vector $|\varphi\rangle \in H$,
$\langle\varphi|\varrho|\varphi\rangle = \sum_k p_k
|\langle\varphi|\psi_k\rangle|^2 \geq 0$. Viceversa, if 
$\varrho$ is a positive operator with unit trace than it can be
diagonalized and the sum of eigenvalues is equal to one. Thus it can be
naturally associated to an ensemble. \qed
\end{proof}
As it is true for any operator, the density operator may be expressed in terms
of its matrix elements in a given basis, i.e. $\varrho=\sum_{np}
\varrho_{np} |\varphi_n\rangle\langle\varphi_p|$ where
$\varrho_{np}=\langle\varphi_n|\varrho|\varphi_p\rangle$ is usually
referred to as the {\em density matrix} of the system.
Of course, the density matrix of a state is diagonal if we use
a basis which coincides or includes the set of eigenvectors of the
density operator, otherwise it contains off-diagonal elements. 
\par
Different ensembles may lead to the same density operator. In this case 
they have the same expectation values for any operator and thus are physically
indistinguishable. In other words, 
different ensembles leading to the same density operator
are actually the same state, i.e. the density operator provides the natural 
and most fundamental quantum description of physical systems.
How this reconciles with Postulate \ref{pqs} dictating that 
{\em physical systems are described by vectors in a Hilbert space}? 
\par  
In order to see how it works let us first notice that, according to the 
postulates reported above, the action of "measuring nothing" should be described by 
the identity operator $\id$. Indeed the identity it is Hermitian and has the 
single eigenvalues $1$, corresponding to the persistent result of measuring 
nothing. Besides, the eigenprojector corresponding to the eigenvalue $1$ is 
the projector over the whole Hilbert space and thus we have the
consistent prediction that the state after the "measurement" is
left unchanged. Let us now consider a situation in which a bipartite system
prepared in the state $|\psi_\smab\rangle\rangle \in H_{\sma} \otimes H_{\smb}$
is subjected to the measurement of an observable $X=\sum_x P_x \in L(H_\sma)
$, $P_x=|x\rangle\langle x|$ i.e. a measurement involving only the degree 
of freedom described by the Hilbert space $H_\sma$. The overall observable
measured on the global system is thus $\boldsymbol{X}=X\otimes \id_\smb$, with 
spectral decomposition  $\boldsymbol{X}= \sum_x x\, \boldsymbol{Q}_x$, 
$\boldsymbol{Q}_x=P_x\otimes \id_\smb$. The probability distribution 
of the outcomes is then obtained using the Born rule, i.e.
\begin{align}
p_x = \hbox{Tr}_\smab
\Big[|\psi_\smab\rangle\rangle\langle\langle\psi_\smab |\,
P_x \otimes \id_\smb\Big]\,. \label{b1} 
\end{align}
On the other hand, since the measurement has been performed on the sole 
system $A$, one expects the Born rule to be valid also at the level of 
the single system $A$, and a question arises on the form of the object 
$\varrho_\sma$ which allows one to write  
$p_x = \hbox{Tr}_\sma \left[\varrho_\sma\, P_x\right]$
i.e. the Born rule as a trace only over the Hilbert space $H_\sma$. 
Upon inspecting Eq. (\ref{b1}) one sees that a suitable  mapping  
$|\psi_\smab\rangle\rangle\langle\langle\psi_\smab | \rightarrow \varrho_\sma$ is
provided by the partial trace 
$\varrho_\sma=\hbox{Tr}_\smb\big[|\psi_\smab\rangle\rangle\langle
\langle\psi_\smab |\big]$. Indeed, for the operator $\varrho_\sma$ 
defined as the partial trace, we have $\hbox{Tr}_\sma[\varrho_\sma]=
\hbox{Tr}_\smab\left[|\psi_\smab\rangle\rangle\langle
\langle\psi_\smab |\right]=1$ and, for any vector $|\varphi\rangle\in
H_\sma$ , 
$\langle\varphi_\sma|\varrho_\sma|\varphi_\sma\rangle = \hbox{Tr}_\smab
\left[|\psi_\smab\rangle\rangle\langle
\langle\psi_\smab |\,
|\varphi_\sma\rangle\langle\varphi_\sma|\otimes
\id_\smb\right] \geq 0$. Being a positive, unit trace, operator
$\varrho_\sma$ is itself a density operator according to Theorem 1. 
As a matter of fact, the partial trace is the unique
operation which allows to maintain the Born rule at both levels, i.e. the
unique operation leading to the correct description of observable
quantities for subsystems of a composite system. Let us state this as a
little theorem \cite{nie00} 
\begin{theorem}[Partial trace] 
The unique mapping 
$|\psi_\smab\rangle\rangle\langle\langle\psi_\smab | \rightarrow \varrho_\sma =
f(\psi_\smab)$ from $H_\sma \otimes H_\smb$ to $H_\sma$ for which  
$\Tr_\smab
\left[|\psi_\smab\rangle\rangle\langle\langle\psi_\smab |\,
P_x \otimes \id_\smb\right] = \Tr_\sma \left[f(\psi_\smab)\,
P_x\right]$ is the partial trace $f(\psi_\smab)\equiv \varrho_\sma = 
\Tr_\smb\left[|\psi_\smab\rangle\rangle\langle
\langle\psi_\smab |\right]$. 
\end{theorem}
\begin{proof} Basically the proof reduces to the fact that the 
set of operators on $H_\sma$ is itself a Hilbert space $L(H_\sma)$ 
with scalar product given by $\langle\langle A| B\rangle\rangle = 
\hbox{Tr}[A^\dag B]$. If we consider a basis of operators
$\{M_k\}$ for $L(H_\sma)$ and expand $f(\psi_\smab) =\sum_k M_k
\hbox{Tr}_\sma[M_k^\dag f(\psi_\smab)]$, then since the map $f$ has to
preserve the Born rule, we have 
$$
f(\psi_\smab) =\sum_k M_k
\hbox{Tr}_\sma[M_k^\dag\, f(\psi_\smab)]
= \sum_k M_k
\hbox{Tr}_\smab\left[M_k^\dag\otimes\id_\smb\,
|\psi_\smab\rangle\rangle\langle\langle\psi_\smab |\right]\,
$$
and the thesis follows from the fact that in a Hilbert space the
decomposition on a basis is unique. \qed
\end{proof}
The above result can be easily generalized to the case of a system 
which is initially described by a density operator $\varrho_\smab$, 
and thus we conclude that when we focus attention to a subsystem of 
a composite larger system the unique mathematical description of 
the act of ignoring part of the degrees of freedom is provided by the partial trace. 
It remains to be proved that the partial trace of a density operator
is a density operator too. This is a very consequence of the definition
that we put in the form of another little theorem. 
\begin{theorem} The partial traces 
$\varrho_\sma = \Tr_\smb[\varrho_\smab]$, 
$\varrho_\smb = \Tr_\sma[\varrho_\smab]$ 
of a density operator $\varrho_\smab$ of a bipartite system, 
are themselves density operators for the reduced systems.
\end{theorem}
\begin{proof} We have 
$\hbox{Tr}_\sma [\varrho_\sma] = \hbox{Tr}_\smb [\varrho_\smb] =
\hbox{Tr}_\smab[\varrho_\smab]=1$ 
and, for any state 
$|\varphi_\sma\rangle\in H_\sma$, $|\varphi_\smb\rangle\in H_\smb$,
\begin{align}
\langle\varphi_\sma|\varrho_\sma|\varphi_\sma\rangle &= \hbox{Tr}_\smab
\left[\varrho_\smab\, |\varphi_\sma\rangle\langle\varphi_\sma|\otimes
\id_\smb\right] \geq 0 \notag \\
\langle\varphi_\smb|\varrho_\smb|\varphi_\smb\rangle &= \hbox{Tr}_\smab
\left[\varrho_\smab\,\id_\sma \otimes
|\varphi_\smb\rangle\langle\varphi_\smb|
\right] \geq 0\,. \notag \quad \qed
\end{align} \end{proof}
\subsubsection{Conditional states}
From the above results it also follows that when we perform a measurement on one of the two
subsystems, the state of the "unmeasured" subsystem after 
the observation of a specific outcome may be obtained as the partial
trace of the overall post measurement state, i.e. the  projection of the 
state before the measurement on the  eigenspace of the observed eigenvalue,
in formula
\begin{align}
\varrho_{\smb x} = \frac{1}{p_x} \hbox{Tr}_\sma\left[ P_x\otimes\id_\smb
\,\varrho_\smab\, P_x\otimes\id_\smb \right]
= \frac{1}{p_x} \hbox{Tr}_\sma\left[\varrho_\smab\, P_x\otimes\id_\smb
\right]\,\label{conditional}
\end{align}
where, in order to write the second equality, we made use of the circularity of
the trace (see Appendix \ref{apTR}) and of the fact that we are dealing with 
a factorized projector. The state $\varrho_{\smb x}$ will be also referred to 
as the "conditional state" of
system $B$ after the observation of the outcome $x$ from a measurement 
of the observable $X$ performed on the system $A$.
\begin{exercise}
Consider a bidimensional system (say the spin state of
a spin $\frac12$ particle) and find two ensembles corresponding to the 
same density operator.
\end{exercise}
\begin{exercise}
Consider a spin $\frac12$ system and the ensemble
$\{p_k,|\psi_k\}$, $k=0,1$, $p_0=p_1=\frac12$,
$|\psi_0\rangle=|0\rangle$, $|\psi_1\rangle=|1\rangle$, where $|k\rangle$
are the eigenstates of $\sigma_3$. Write the density matrix in the basis
made of the eigenstates of $\sigma_3$ and then in the basis of $\sigma_1$. Then, do the same but
for the ensemble obtained from the previous one by
changing the probabilities to $p_0=\frac14$, $p_1=\frac34$.
\end{exercise}
\begin{exercise}
Write down the partial traces of the state
$|\psi\rangle\rangle=\cos\phi\, |00\rangle\rangle + \sin\phi\,
|11\rangle\rangle$, 
where we used the notation $|jk\rangle\rangle=|j\rangle\otimes
|k\rangle$.
\end{exercise}
\subsection{Purity and purification of a mixed state}
As we have seen in the previous section when we observe a portion,
say $A$, of a composite system described by the vector
$|\psi_\smab\rangle\rangle\in
H_\sma \otimes H_\smb$, the mathematical object to be inserted in the
Born rule in order to have the correct description of observable
quantities is the partial trace, which individuates a density operator on
$H_\sma$. Actually, also the converse is true, i.e. any density operator
on a given Hilbert space may be viewed as the partial trace of a state
vector on a larger Hilbert space. Let us prove this constructively: 
if $\varrho$ is a density operator on $H$, then it can be diagonalized by
its eigenvectors and it can be written as $\varrho=\sum_k \lambda_k
|\psi_k\rangle\langle\psi_k|$; then we introduce another Hilbert space 
$K$, with dimension at least equal to the number of nonzero eqigenvalues
of $\varrho$ and a basis
$\{|\theta_k\rangle\}$ in $K$, and consider the vector
$|\varphi\rangle\rangle \in H \otimes K$ given by
$|\varphi\rangle\rangle=\sum_k \sqrt{\lambda_k}\, |\psi_k\rangle \otimes
|\theta_k\rangle$. Upon tracing over the Hilbert space $K$, we have
$$
\hbox{Tr}_\smk \left[|\varphi\rangle\rangle\langle\langle\varphi|\right] = 
\sum_{kk^\prime} \sqrt{\lambda_k\lambda_{k^\prime}}\,
|\psi_k\rangle\langle\psi_{k^\prime}|\,
\langle\theta_{k^\prime}|\theta_k\rangle= \sum_k \lambda_k\,
|\psi_k\rangle\langle\psi_k| = \varrho\:.
$$
Any vector on a larger Hilbert space which satisfies the above condition 
is referred to as a {\em purification} of the given density operator. Notice
that, as it is apparent from the proof, there exist infinite purifications 
of a density operator. Overall, putting together this fact with the conclusions
from the previous section, we are led to reformulate the first postulate 
to say that  {\em quantum states of a physical system are described
by density operators}, i.e. positive operators with unit trace 
on the Hilbert space of the system.
\par
A suitable measure to quantify how far a density operator is from a
projector is the so-called {\em purity}, which is defined
as the trace of the square density operator $\mu[\varrho]=
\hbox{Tr}[\varrho^2]=\sum_k \lambda_k^2$, where the $\lambda_k$'s are 
the eigenvalues of $\varrho$. Density operators made by a projector
$\varrho=|\psi\rangle\langle\psi|$ have $\mu=1$ and are referred to
as {\em pure states}, whereas for any $\mu<1$ we have a
{\em mixed state}. Purity of a state ranges in the interval
$1/d \leq \mu \leq 1$ where $d$ is the dimension of the Hilbert space.
The lower bound is found looking for the minimum of $\mu=\sum_k \lambda_k^2$
with the constraint $\sum_k \lambda_k=1$, and amounts to minimize
the function $F=\mu+\gamma\sum_k \lambda_k$, $\gamma$ being a Lagrange
multipliers. The solution is $\lambda_k=1/d$, $\forall k$, i.e.
the {maximally mixed} state $\varrho=\id/d$, and the
corresponding purity is $\mu=1/d$. 
\par
When a system is prepared in a pure state we have 
the maximum possible information on the system according 
to quantum mechanics. On the other hand, for mixed states 
the degree of purity is connected with the amount of information 
we are missing by looking at the system only, while ignoring the 
{\em environment}, i.e. the rest of the universe. In fact, 
by looking at a portion of a composite system we are ignoring
the information encoded in the correlations between the portion under
investigation and
the rest of system: This results in a smaller amount of information 
about the state of the subsystem itself. In order to emphasize this
aspect, i.e. the existence of residual ignorance about the system, 
the degree of mixedness 
may be quantified also by the Von Neumann (VN) entropy 
$S[\varrho]=-\hbox{Tr}\left[\varrho\,\log\varrho\right]=-\sum_n
\lambda_n \log\lambda_n$, where $\{\lambda_n\}$ are the eigenvalues of 
$\varrho$. We have $0\leq S[\varrho] \leq \log d$: for a pure state 
$S[|\psi\rangle\langle\psi|]=0$ whereas $S[\id/d]=\log d$ for a 
maximally mixed state. VN entropy is a monotone function of the purity,
and viceversa.
\begin{exercise}
Evaluate purity and VN entropy of the 
partial traces of the state $|\psi\rangle\rangle=\cos\phi\, 
|01\rangle\rangle + \sin\phi\, |10\rangle\rangle$.
\end{exercise}
\begin{exercise}
Prove that for any pure bipartite state the entropies
of the partial traces are equal, though the two density
operators need not to be equal.
\end{exercise}
\begin{exercise}
Take a single-qubit state with density operator expressed in terms of
the Pauli matrices $\varrho=\frac12 (\id + r_1 \sigma_1 + r_2 \sigma_2+ r_3
\sigma_3)$ (Bloch sphere representation), $r_k=\Tr[\varrho\, \sigma_k]$, and prove that the Bloch
vector $(r_1,r_2,r_3)$ should satisfies $r_1^2+r_2^2+r_3^3\leq 1$ 
for $\varrho$ to be a density
operator.
\end{exercise}
\section{Quantum measurements}
In this section we put the postulates of standard quantum measurement 
theory under closer scrutiny. We start with some formal considerations 
and end up with a reformulation suitable for the description of any
measurement performed on a quantum system, including those involving
external systems or a noisy environment \cite{Per93,Bergou}. \par 
Let us start by reviewing the
postulate of standard quantum measurement theory in a
pedantic way, i.e. by expanding Postulate \ref{pqm};   
$\varrho$ denotes the state of the system before the
measurement.
\begin{itemize}
\item[{\bf [2.1]}] Any observable quantity is associated to a
Hermitian operator $X$ with spectral decomposition $X=\sum_x
\,x\,|x\rangle\langle x|$. The eigenvalues are real and we assume for
simplicity that they are nondegenerate. A measurement of $X$ yields 
one of the eigenvalues $x$ as possible outcomes.
\item[{\bf [2.2]}] The eigenvectors of $X$ form a basis
for the Hilbert space. The projectors $P_x=|x\rangle\langle x|$ span
the entire Hilbert space, $\sum_x P_x=\id$.
\item[{\bf [2.3]}] The projectors $P_x$ are orthogonal $P_xP_{x^\prime} =
\delta_{xx^\prime}P_x$. It follows that $P_x^2=P_x$ and thus that the
eigenvalues of any projector are $0$ and $1$.
\item[{\bf [2.4]}] { (Born rule)} The probability that a particular outcome
is found as the measurement result is 
$$p_x=  \hbox{Tr}\left[P_x\varrho P_x\right] = \hbox{Tr}\left[\varrho
P_x^2\right] \stackrel{\bigstar}{=} \hbox{Tr}\left[\varrho P_x\right]\,.$$
\item[{\bf [2.5]}] { (Reduction rule)} The state after the measurement (reduction
rule or projection postulate) is 
$$\varrho_x = \frac1{p_x}\,P_x\varrho P_x\,,$$ if the outcome is $x$.
\item[{\bf [2.6]}] If we perform a measurement but we 
do not record the results, the post-measurement state is given by 
$\widetilde{\varrho}=\sum_x p_x\,\varrho_x = \sum_x P_x\varrho P_x$.
\end{itemize}
The formulations {\bf [2.4]} and ${\bf [2.5]}$ follow from the
formulations 
for pure states, upon invoking the existence of a
purification:
\begin{align} 
p_x &= \hbox{Tr}_\smab \left[P_x \otimes \id_\smb\,
|\psi_\smab\rangle\rangle\langle\langle\psi_\smab |\,P_x \otimes
\id_\smb \right]=
\hbox{Tr}_\smab \left[
|\psi_\smab\rangle\rangle\langle\langle\psi_\smab |\,P_x^2 \otimes
\id_\smb \right] \notag \\ &=\hbox{Tr}_\sma\left[\varrho_\sma
P_x^2\right] \,
\end{align}
\begin{align} \varrho_{\sma x} &= 
\frac1{p_x} \hbox{Tr}_\smb \left[
P_x \otimes\id_\smb\,
|\psi_\smab\rangle\rangle\langle\langle\psi_\smab |\,P_x \otimes
\id_\smb\right]=
\frac1{p_x}P_x \,
 \hbox{Tr}_\smb \left[
|\psi_\smab\rangle\rangle\langle\langle\psi_\smab |\right] P_x
\notag \\ &=\frac1{p_x}P_x \, \varrho_\sma\,P_x\,.
\end{align}
The message conveyed by these postulates is that we can only 
predict the spectrum of the possible outcomes and the 
probability that a given outcome is obtained. On the other hand, 
the measurement process is random, and we cannot predict the actual 
outcome of each run. Independently on its purity, a density operator $\varrho$ does
not describe the state of a single system, but rather an ensemble of identically
prepared systems. If we perform the same measurement on each member
of the ensemble we can predict the possible results and the probability
with which they occur but we cannot predict the result of individual
measurement (except when the probability of a certain outcome is either
$0$ or $1$).
\subsection{Probability operator-valued measure and detection operators}
The set of postulates {\bf [2.*]} may be seen as a set of recipes
to generate probabilities and post-measurement states. 
We also notice that the number of
possible outcomes is limited by the number of terms in the orthogonal 
resolution of identity, which itself cannot be larger than the dimensionality
of the Hilbert space. It would however be often desirable to have more
outcomes than the dimension of the Hilbert space while keeping
positivity and normalization of probability distributions. In this
section will show that this is formally possible, upon relaxing the
assumptions on the mathematical objects describing the measurement,
and replacing them with more flexible ones, still obtaining a meaningful
prescription to generate probabilities. Then, in the next sections we
will show that there are physical processes that fit with this 
generalized description, and that actually no revision of the postulates
is needed, provided that the degrees of freedom of the measurement
apparatus are taken into account.
\par
The Born rule is a prescription to generate probabilities: its textbook 
form is the right term of the starred equality in ${\bf [2.4]}$.
However, the form on the left term has the merit to underline that in order to 
generate a probability it sufficient if the $P_x^2$ is a positive
operator. In fact, we do not need to require that the set of the $P_x$'s 
are projectors, nor we need the
positivity of the underlying $P_x$ operators. So, let us consider the 
following generalization: we introduce a set of positive operators 
$\Pi_x\geq 0$, which are the generalization of the $P_x$ and use the 
prescription $p_x=\hbox{Tr}[\varrho\,\Pi_x]$ to generate probabilities.
Of course, we want to ensure that this is a true probability 
distribution, i.e. normalized, and therefore require that 
$\sum_x \Pi_x= \id$, that is the positive operators 
still represent a resolution of the identity, as the set of projectors over 
the eigenstates of a selfadjoint operator. We will call a decomposition of 
the identity in terms of positive operators $\sum_x \Pi_x=\id$ a 
{\em probability operator-valued measure} (POVM) 
and $\Pi_x\geq0$ the elements of the POVM.
\par
Let us denote the operators giving the post-measurement states 
(as in ${\bf [2.5]}$) by $M_x$. We refer to them as to the {\em detection
operators}. As noted above, they  are no longer constrained to be 
projectors. Actually, they may be any operator 
with the constraint, imposed by ${\bf [2.4]}$ i.e.
$p_x = \hbox{Tr}[M_x\varrho\, M_x^\dag] = \hbox{Tr}[\varrho\, \Pi_x]$.
This tells us that the POVM elements have the form $\Pi_x=M_x^\dag M_x$ 
which, by construction, individuate a set of 
a positive operators. There is a residual
freedom in designing the post-measurement state. In fact, since 
$\Pi_x$ is a positive operator $M_x=\sqrt{\Pi_x}$ exists and satisfies
the constraint, as well as any operator of the form  
$M_x=U_x\,\sqrt{\Pi_x}$ with $U_x$ unitary. This is the most general 
form of the detection operators satisfying the constraint 
$\Pi_x=M_x^\dag M_x$ and corresponds to their polar decomposition.
The POVM elements determine the absolute values leaving the freedom of 
choosing the unitary part. 
\par
Overall, the detection operators $M_x$ represent a generalization of the
projectors $P_x$, while the POVM elements $\Pi_x$ generalize $P_x^2$.
The postulates for quantum measurements may be reformulated as follows
\begin{itemize}
\item[{\bf [II.1]}]  Observable quantities are associated to POVMs, 
i.e. decompositions of identity $\sum_x \Pi_x = \id$ in terms
of positive $\Pi_x\geq 0$ operators. The possible outcomes $x$ label
the elements of the POVM and the construction may be generalized to the
continuous spectrum.
\item[{\bf [II.2]}] The elements of a POVM are positive
operators expressible as $\Pi_x=M^\dag_x\,M_x$ where the detection
operators $M_x$ are generic operators with the only constraint 
$\sum_x M^\dag_x\,M_x = \id$.
\item[{\bf [II.3]}] (Born rule) The probability that a particular outcome
is found as the measurement result is
$p_x=  \hbox{Tr}\left[M_x\varrho M_x^\dag\right] = \hbox{Tr}\left[\varrho
M_x^\dag M_x\right] = \hbox{Tr}\left[\varrho \Pi_x\right]$.
\item[{\bf [II.4]}] (Reduction rule) The state after the measurement
is 
$\varrho_x = \frac1{p_x}\,M_x\varrho M_x^\dag$ if the outcome is $x$.
\item[{\bf [II.5]}] If we perform a measurement but we 
do not record the results, the post-measurement state is given by 
$\widetilde{\varrho}=\sum_x p_x\,\varrho_x = \sum_x M_x\varrho M_x^\dag$.
\end{itemize}
Since orthogonality is no longer a requirement, the number of elements 
of a POVM has no restrictions and so the number of possible outcomes
from the measurement. The above formulation generalizes both the 
Born rule and the reduction rule, and says that any set of
detection operators satisfying ${\bf [II.2]}$ corresponds to 
a legitimate operations leading to a proper probability distribution 
and to a set of post-measurement states. This scheme is referred to as 
a {\em generalized measurement}. Notice that in ${\bf [II.4]}$
we assume a reduction mechanism sending pure states into pure states.
This may be further generalized to reduction mechanism where pure
states are transformed to mixtures, but we are not going to deal with
this point.
\par
Of course, up to this point, this is just a formal mathematical 
generalization of the standard description of measurements given 
in  textbook quantum mechanics, and few questions naturally arise:  
Do generalized measurements describe physically realizable measurements?  
How they can be implemented? And if this is the case, does it 
means that standard formulation is too restrictive or wrong? To all these 
questions an answer will be provided by the following sections where 
we state and prove the Naimark Theorem, and discuss few examples of
measurements described by POVMs.
\subsection{The Naimark theorem}
The Naimark theorem basically says that any generalized measurement
satisfying {\bf [II.*]} may be viewed as a standard measurement 
defined by {\bf [2.*]} in a larger Hilbert space, and 
conversely, any standard measurement involving more than one physical
system may be described as a generalized measurement on one of the
subsystems. In other words,
if we focus attention on a portion of a composite system
where a standard measurement takes place, than the statistics of the
outcomes and the post-measurement states of the subsystem may be 
obtained with the tools of generalized measurements. Overall, we
have
\begin{theorem}[Naimark] For any given {\rm POVM} $\sum_x \Pi_x = \id$,
$\Pi_x\geq 0$ on a Hilbert space $H_\sma$ there exists a Hilbert space 
$H_\smb$, a state $\varrho_\smb=|\omega_\smb\rangle\langle\omega_\smb|
\in L(H_\smb)$, a unitary operation $U\in L(H_\sma \otimes H_\smb)$,
$UU^\dag=U^\dag U=\id$, and a projective measurement $P_x$,
$P_xP_x^\prime=\delta_{xx^\prime} P_x$ on $H_\smb$ such that 
$\Pi_x=\Tr_\smb [\id \otimes \varrho_\smb\, U^\dag \id \otimes
P_x\,U]$. The setup is referred to as a {\em Naimark extension} of
the {\rm POVM}. Conversely, any measurement scheme where the system is 
coupled to another system, from now on referred to as the  ancilla, and 
after evolution, a projective measurement 
is performed on the ancilla may be seen as the Naimark extension of a
{\rm POVM}, i.e. one may write the Born rule
$p_x=\Tr[\varrho_\sma\,\Pi_x]$ and the reduction rule
$\varrho_\sma \rightarrow \varrho_{\sma x}=\frac1{p_x}M_x\varrho_\sma
M_x^\dag$
at the level of the system only, in terms of the POVM elements
$\Pi_x=\Tr_\smb [\id \otimes \varrho_\smb\, U^\dag \id \otimes
P_x\,U]$ and the detection operators $M_x|\varphi_\sma\rangle = \langle
x|U|\varphi_\sma,\omega_\smb\rangle\rangle$.
\end{theorem}
Let us start with the second part of the theorem, and look 
at what happens when we couple the system 
under investigation to an additional system, usually referred to as ancilla
(or apparatus), let them evolve, and then perform a projective
measurement on the ancilla. This kind of setup is schematically depicted
in Figure 1. 
\begin{figure}[h!]
\centerline{\includegraphics[width=0.85\textwidth]{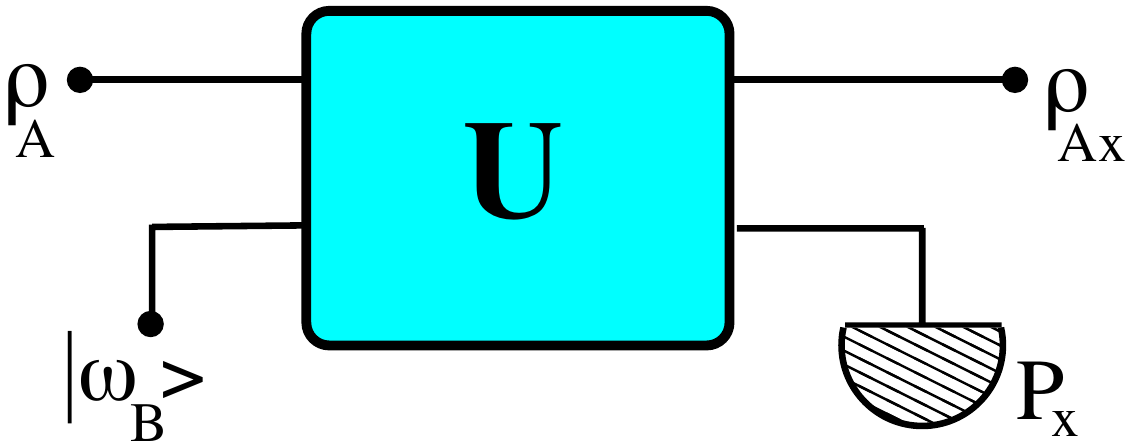}}
\caption{Schematic diagram of a generalized measurement.
The system of interest is coupled to an ancilla prepared in a known
state $|\omega_\smb\rangle$ by the unitary evolution $U$, and then 
a projective measurement is performed on the ancilla.}
\end{figure}
\par
The Hilbert space of the overall system is $H_\sma\otimes
H_\smb$, and we assume that the system and the ancilla are initially
independent on each other, i.e. the global initial preparation is
$R=\varrho_\sma\otimes
\varrho_\smb$. We also assume that the ancilla is prepared in the pure state
$\varrho_\smb=|\omega_\smb\rangle\langle\omega_\smb|$ since this 
is always possible, upon a suitable purification of the ancilla degrees of
freedom, i.e. by suitably enlarging the ancilla Hilbert space.  Our aim
it to obtain information about the system by measuring an observable $X$ on 
the ancilla. This is done after the system-ancilla  interaction 
described by the unitary operation $U$.
According to the Born rule the probability of the outcomes is given by
$$
p_x = \hbox{Tr}_\smab \left[U \varrho_\sma\otimes\varrho_\smb U^\dag\,
\id\otimes |x\rangle\langle x|\right] = 
\hbox{Tr}_\sma\left[\varrho_\sma\,\underbrace{
\hbox{Tr}_\smb\left[
\id \otimes \varrho_\smb\, 
U^\dag\,\id\otimes |x\rangle\langle x|U\right]}
\right]
\vspace{-2mm}
$$
\begin{flushright}{\mbox{\large${\;\Pi_x}$}$\qquad\qquad
\qquad\qquad$}\end{flushright}
where the set of operators $\Pi_x=\hbox{Tr}_\smb\left[
\id \otimes \varrho_\smb\, U^\dag\,\id\otimes |x\rangle\langle x|U\right]
= \langle\omega_\smb| U^\dag \id \otimes P_x U | \omega_\smb\rangle $ 
is the object that would permit to
write the Born rule at the level of the subsystem $A$, i.e. it is our
candidate POVM.  
\par
In order to prove this, let us define
the operators $M_x\in L(H_\sma)$ by their action on the generic vector
in $H_\sma$
$$ 
M_x |\varphi_\sma\rangle = \langle x| U 
| \varphi_\sma,\omega_\smb\rangle\rangle
$$
where 
$| \varphi_\sma,\omega_\smb\rangle\rangle = | \varphi_\sma\rangle\otimes
|\omega_\smb\rangle$ and the $|x\rangle$'s are the orthogonal 
eigenvectors of $X$. Using the decomposition of $\varrho_\sma=\sum_k
\lambda_k |\psi_k\rangle\langle\psi_k |$ onto its eigenvectors the 
probability of the outcomes can be rewritten as
\begin{align}\notag
p_x &= \hbox{Tr}_\smab \left[U \varrho_\sma\otimes\varrho_\smb U^\dag\,
\id\otimes |x\rangle\langle x|\right] = \sum_k \lambda_k \hbox{Tr}_\smab
\left[U
|\psi_k,\omega_\smb\rangle\rangle\langle\langle\omega_\smb,\psi_k |
U^\dag\, \id\otimes |x\rangle\langle x| \right]\\ \notag
& = 
\sum_k \lambda_k \hbox{Tr}_\sma
\left[\langle x|U
|\psi_k,\omega_\smb\rangle\rangle\langle\langle\omega_\smb,\psi_k |
U^\dag|x\rangle\right]
= \sum_k \lambda_k \hbox{Tr}_\sma
\left[M_x|\psi_k\rangle\langle\psi_k |
M_x^\dag\right] \\ 
 &= \hbox{Tr}_\sma \left[M_x \varrho_\sma M_x^\dag\right]
= \hbox{Tr}_\sma \left[\varrho_\sma\, M_x^\dag M_x\right]\,,
\label{bigstar1} 
\end{align}
which shows that $\Pi_x=M_x^\dag M_x$ is indeed a positive operator
$\forall x$.
Besides, for any vector $|\varphi_\sma\rangle$ in $H_\sma$ we have
\begin{align}
\langle\varphi_\sma| \sum_x M^\dag_x M_x |\varphi_\sma\rangle &= 
\sum_x \langle\langle \omega_\smb,\varphi_\sma|U^\dag |x\rangle\langle x|
U|\varphi_\sma, \omega_\smb\rangle\rangle \notag \\
&=\langle\langle \omega_\smb,\varphi_\sma|U^\dag
U|\varphi_\sma, \omega_\smb\rangle\rangle=1\,,
\label{bigstar2} 
\end{align}
and since this is true for any $|\varphi_\sma\rangle$ we have $\sum_x M_x^\dag
M_x=\id$. Putting together Eqs. (\ref{bigstar1}) and (\ref{bigstar2}) we have 
that the set of operators $\Pi_x=M^\dag_x M_x$ is a POVM, with detection operators
$M_x$. In turn, the conditional state of the system $A$, after having
observed the outcome $x$, is given by
\begin{align}
\varrho_{\sma x} &= \frac1{p_x} \hbox{Tr}_\smb \left[U \varrho_\sma
\otimes |\omega_\smb\rangle\langle\omega_\smb | U^\dag\, \id\otimes P_x \right]
= \frac1{p_x}\sum_k \lambda_k \langle
x|U|\psi_k,\omega_\smb\rangle\rangle\langle\langle\omega_\smb,\psi_k
|U^\dag |x\rangle \notag \\ &= \frac1{p_x} M_x \varrho_\sma M_x^\dag
\end{align}
This is the half of the Naimark theorem: if we couple our system to an
ancilla, let them evolve and perform the measurement of an observable 
on the ancilla, which projects the ancilla on a basis in $H_\smb$, then
this procedure also modify the system. The transformation needs not to
be a projection. Rather, it is adequately described by a
set of detection operators which realizes a POVM on the system Hilbert
space. Overall, the meaning of the above proof is twofold: on the one hand 
we have shown that there exists realistic measurement schemes which are
described by POVMs when we look at the system only. At the same time, we
have shown that the partial trace of a spectral measure is a POVM, which
itself depends on the projective measurement performed on the
ancilla, and on its initial preparation. Finally, we notice that the 
scheme  of Figure 1 provides a general model for any kind of detector 
with internal degrees of freedom.
\par
Let us now address the converse problem: given a set of detection
operators $M_x$ which realizes a POVM $\sum_x M^\dag_x M_x=\id$, is 
this the system-only description of an indirect measurement 
performed a larger Hilbert space? In other words, there exists 
a Hilbert space $H_\smb$, a state $\varrho_\smb=
|\omega_\smb\rangle\langle\omega_\smb|\in L(H_\smb)$, a unitary 
$U\in L(H_\sma \otimes H_\smb)$, and a projective measurement 
$P_x=|x\rangle\langle x|$ in $H_\smb$ such that 
$M_x |\varphi_\sma\rangle = \langle x| U 
| \varphi_\sma,\omega_\smb\rangle\rangle$ holds for any
$|\varphi_\sma\rangle \in H_\sma$ and $\Pi_x= \langle\omega_\smb| 
U^\dag \id \otimes P_x U | \omega_\smb\rangle $? 
The answer is positive and we will provide a constructive proof.
Let us take $H_\smb$ with dimension equal to the number of
detection operators and of POVM elements, and choose a basis $|x\rangle$
for $H_\smb$, which in turn individuates a projective measurement.
Then we choose an arbitrary state $|\omega_\smb\rangle \in H_\smb$
and define the action of an operator U as
$$
U\,|\varphi_\sma\rangle \otimes |\omega_\smb \rangle = \sum_x
M_x\,|\varphi_\sma\rangle \otimes |x\rangle
$$
where $|\varphi_\sma\rangle \in H_\sma$ is arbitrary.
The operator $U$ preserves the scalar product
\begin{align}\notag
\langle\langle \omega_\smb,\varphi_\sma^\prime | U^\dag U|
\varphi_\sma,\omega_\smb \rangle\rangle 
= \sum_{x x^\prime}
\langle \varphi_\sma^\prime | M_{x^\prime}^\dag
M_x|\varphi_\sma\rangle \langle x^\prime|x\rangle
= \sum_{x}
\langle\varphi_\sma^\prime | M_{x^\prime}^\dag
M_x|\varphi_\sma\rangle
= 
\langle\varphi_\sma^\prime | \varphi_\sma\rangle
\end{align}
and so it is unitary in the one-dimensional subspace spanned by 
$|\omega_\smb\rangle$. Besides, it may be extended to a full 
unitary operator in the global Hilbert space $H_\sma\otimes H_\smb$, 
eg it can be the identity operator in the subspace orthogonal to
$|\omega_\smb\rangle$. Finally, for any $|\varphi_\sma\rangle\in
H_\sma$, we have 
$$\langle x|U|\varphi_\sma,\omega_\smb\rangle\rangle = 
\sum_{x^\prime} M_{x^\prime}|\varphi_\sma\rangle \langle
x|x^\prime\rangle = M_x |\varphi_\sma\rangle\,,$$
and
$$\langle\varphi_\sma|\Pi_x|\varphi_\sma\rangle=
\langle\varphi_\sma|M_x^\dag M_x|\varphi_\sma\rangle=
\langle\langle\omega_\smb,\varphi_\sma|U^\dag \id\otimes P_x
U|\varphi_\sma,\omega_\smb\rangle\rangle\,,$$
that is,
$\Pi_x= \langle\omega_\smb| 
U^\dag \id \otimes P_x U | \omega_\smb\rangle$. \hfill $\qed$
\par
This completes the {proof of the Naimark theorem}, which 
asserts that there is a one-to-one correspondence between POVM and indirect
measurements of the type describe above. In other words, an indirect
measurement may be seen as the physical implementation of a POVM and
any POVM may be realized by an indirect measurement. 
\par
The emerging picture is thus the following: In measuring a quantity of
interest on a physical system one generally deals with a larger system
that involves additional degrees of freedom, besides those of the system 
itself. These additional physical entities are globally referred to as the
apparatus or the ancilla. As a matter of fact, the measured quantity may be
always described by a standard observable, however on a larger
Hilbert space describing both the system and the apparatus. When we
trace out the degrees of freedom of the apparatus we are generally 
left with a POVM rather than a PVM. Conversely, any conceivable POVM,
i.e. a set of positive operators providing a resolution of identity, 
describe a generalized measurement, which may be always implemented 
as a standard measurement in a larger Hilbert space.
\par Before ending this Section, few remarks are in order:
\begin{itemize}
\item[{\bf R1}] The possible Naimark extensions are actually infinite, 
corresponding to the intuitive idea that there are infinite ways, with 
an arbitrary  number of ancillary systems, of measuring a given quantity. 
The construction reported above is sometimes referred to as the {\em 
canonical extension} of a POVM. The Naimark theorem just says that an 
implementation in terms of an ancilla-based indirect measurement is always possible,
but of course the actual implementation may be
different from the canonical one.
\item[{\bf R2}] The projection postulate described at the beginning of
this section, the scheme of indirect measurement, and the canonical 
extension of a POVM have in common the assumption that a
nondemolitive detection scheme takes place, in which the system after 
the measurement has been modified, but still exists. This is sometimes
referred to as a {\em measurement of the first kind} in textbook
quantum mechanics. Conversely, in a demolitive measurement or 
{\em measurement of the second kind}, 
the system is destroyed during the measurement and it makes 
no sense of speaking of the state of the system after the measurement. 
Notice, however, that for demolitive measurements on a field the 
formalism of generalized measurements provides the framework for the 
correct description  of the state evolution. As for example, let us
consider the detection of photons on a single-mode of the radiation field.
A demolitive photodetector (as those based on the absorption of light) 
realizes, in ideal condition, the measurement of the number operator 
$a^\dag a$ without leaving any photon in the mode .
If $\varrho=\sum_{np} \varrho_{np}|n\rangle\langle p|$ is the state
of the single-mode radiation field a photodetector of this kind gives
a natural number $n$ as output, with probability $p_n=\varrho_{nn}$,
whereas the post-measurement state is the vacuum $|0\rangle\langle 0|$
independently on the outcome of the measurement. This kind of
measurement is described by the orthogonal POVM $\Pi_n=|n\rangle\langle
n|$, made by the eigenvectors of the number operator, and by the
detection operator $M_n=|0\rangle\langle n|$. The proof is left as an
exercise.
\item[{\bf R3}] We have formulated and proved the Naimark theorem in a 
restricted form, suitable for our purposes. It should be noticed that it holds in
more general terms, as for example with extension of the Hilbert space
given by direct sum rather than tensor product, and also relaxing the
hypothesis \cite{Pau}.
\subsubsection{Conditional states in generalized measurements}
If we have a
composite system and we perform a projective measurement on, say,
subsystem $A$, the conditional state of the unmeasured subsystem $B$
after the observation of the outcome $x$ is given by Eq.
(\ref{conditional}), i.e. it is the partial
trace of the projection of the state before the measurement on the
eigenspace of the observed eigenvalue. 
One may wonder whether a similar results holds also when the measurement
performed on the subsystem a $A$ is described by a POVM. The answer is
positive and the proof may be given in two ways. The first is based on
the observation that, thanks to the existence of a canonical Naimark
extension, we may write the state of the global system after the
measurement as $$\varrho_{\smab x} = \frac1{p_x} M_x \otimes \id_\smb\,
\varrho_\smab\, M_x^\dag \otimes \id_\smb\,,$$ and thus the conditional
state of subsystem $B$ is the partial trace $\varrho_{\smb
x}=\hbox{Tr}_\sma [\varrho_{\smab x}]$ i.e.
$$\varrho_{\smb x}=
\frac1{p_x}\hbox{Tr}_\sma [M_x \otimes \id_\smb\,
\varrho_\smab\, M_x^\dag \otimes \id_\smb
]= \frac1{p_x}\hbox{Tr}_\sma [
\varrho_\smab\, M_x^\dag M_x \otimes \id_\smb
]= \frac1{p_x}\hbox{Tr}_\sma [
\varrho_\smab\, \Pi_x \otimes \id_\smb
]\,,$$
where again we used the circularity of partial trace in the presence of
factorized operators. A second proof may be offered invoking the Naimark
theorem only to ensure the existence of an extension, i.e. a projective
measurement on a larger Hilbert space $H_\smc \otimes\ H_\sma$, which
reduces to the POVM after tracing over $H_\smc$. In formula, assuming
that $P_x \in L(H_\smc \otimes\ H_\sma)$ is a projector and $\sigma\in
L(H_\smc)$ a density operator
\begin{align}\notag
\varrho_{\smb x}
&=\frac1{p_x}\hbox{Tr}_{\smc\sma} \left[
P_x\otimes\id_\smb\,
\varrho_\smab \otimes \sigma\,
P_x\otimes\id_\smb
\right]=\frac1{p_x}\hbox{Tr}_{\smc\sma} \left[
\varrho_\smab \otimes \sigma\,
P_x\otimes\id_\smb
\right] \\ \notag
&=\frac1{p_x}\hbox{Tr}_{\sma} \left[
\varrho_\smab
\Pi_x\otimes\id_\smb
\right]\,.
\end{align}
\end{itemize}
\subsection{Joint measurement of non commuting observables}
\label{jm}
A common statement about quantum measurements says that it is
not possible to perform a joint measurement of two observables 
$Q_\sma$ and $P_\sma$ of a given system $A$ if they do not commute, i.e. 
$[Q_\sma,P_\sma]\neq 0$. This is related to the impossibility of finding 
any common set of projectors on the Hilbert space $H_\sma$ of the system 
and to define a joint observable. On the other hand, a question 
arises on whether common projectors may be found in a larger 
Hilbert space, i.e. whether one may implement a joint measurement 
in the form of a generalized measurement. The answer is indeed 
positive \cite{art1,yue82}: This Section is devoted to describe the canonical 
implementation of joint measurements for pair of observables
having a (nonzero) commutator $[Q_\sma,P_\sma]=c\, \id \neq 0$ 
proportional to the identity operator.
\par
The basic idea is to look for a pair of commuting observables 
$[X_\smab,Y_\smab]=0$ in a larger Hilbert space $H_\sma \otimes H_\smb$ 
which {\em trace} the observables $P_\sma$ and $Q_\sma$, i.e. which have
the same expectation values 
\begin{align}
\langle X_\smab\rangle\equiv 
\hbox{Tr}_\smab[X_\smab\, \varrho_\sma\otimes\varrho_\smb]  &=
\hbox{Tr}_\sma [Q_\sma\, \varrho_\sma] 
\equiv \langle Q_\sma \rangle \notag
\\
\langle Y_\smab\rangle\equiv 
\hbox{Tr}_\smab[Y_\smab\, \varrho_\sma\otimes\varrho_\smb]  &= 
\hbox{Tr}_\sma [P_\sma\, \varrho_\sma]
\equiv\langle P_\sma \rangle \label{j1}
\end{align}
for any state $\varrho_\sma \in H_\sma$ of the system under investigation, 
and a fixed suitable preparation $\varrho_\smb\in H_\smb$ of the 
system $B$. A pair of such observables may be found upon choosing  
a replica system $B$, identical to $A$, and considering the operators
\begin{align}
X_\smab & = Q_\sma \otimes \id_\smb + \id_\sma \otimes Q_\smb \notag \\
Y_\smab & = P_\sma \otimes \id_\smb - \id_\sma \otimes P_\smb \label{j2}
\end{align}
where $Q_\smb$ and $P_\smb$ are the analogue of $Q_\sma$ and $P_\sma$
for system $B$, see \cite{BV10} for more details involving the
requirement of covariance.
The operators in Eq. (\ref{j2}), taken together a state 
$\varrho_\smb \in H_\smb$ satisfying 
\begin{align}
\hbox{Tr}_\smb [Q_\smb\, \varrho_\smb] 
=\hbox{Tr}_\smb [P_\smb\, \varrho_\smb]=0\,, 
 \label{j3}
\end{align}
fulfill the conditions in Eq. (\ref{j1}), i.e. realize a joint generalized
measurement of the noncommuting observables $Q_\sma$ and $P_\sma$.
The operators $X_\smab$ and $Y_\smab$ are Hermitian by construction. 
Their commutator is given by
\begin{align}
[X_\smab,Y_\smab]= [Q_\sma,P_\sma] \otimes \id_\smb - 
\id_\sma \otimes [Q_\smb,P_\smb] = 0\,.
 \label{j4}
\end{align}
Notice that the last equality, i.e. the fact that the two operators
commute, is valid only if the commutator $[Q_\sma,P_\sma]=c\, \id$
is proportional to the identity. More general constructions are needed 
if this condition does not hold \cite{jsp1}. 
\par
Since the $[X_\smab,Y_\smab]=0$ the complex operator $Z_\smab = X_\smab 
+ i\, Y_\smab$ is {\em normal} i.e. $[Z_\smab,Z_\smab^\dag]=0$. For
normal operators the spectral theorem holds, and we may write  
\begin{align}
Z_\smab = \sum_z z\, P_z \qquad P_z=\kket{z} \bbra{z}\qquad Z_\smab
\kket{z} = z \kket{z}
\label{j5}
\end{align}
where $z\in {\mathbbm C}$, and $P_z$ are orthogonal projectors
on the eigenstates $\kket{z}\equiv \kket{z}_\smab$ of $Z_\smab$.
The set $\{P_z\}$ represents the common projectors individuating 
the joint observable $Z_\smab$. Each run of the measurement
returns a complex
number, whose real and imaginary parts correspond to
a sample of the $X_\smab$ and $Y_\smab$ values, aiming at sampling 
$Q_\sma$ and $P_\sma$. The statistics 
of the measurement is given by 
\begin{align}
p_\smz(z) = \hbox{Tr}_\smab[\varrho_\sma\otimes\varrho_\smb\, P_z]= 
\hbox{Tr}_\sma[\varrho_\sma\, \Pi_z] 
\label{j6}
\end{align}
where the POVM $\Pi_z$ is given by
\begin{align}
\Pi_z = \hbox{Tr}_\smb[\id_\sma\otimes\varrho_\smb\, P_z]\,.
\label{j7}
\end{align}
The mean values $\langle X_\smab\rangle=\langle Q_\sma\rangle$ and 
$\langle Y_\smab\rangle= \langle P_\sma\rangle$ 
are the correct ones by construction, where by saying "correct" we
intend the mean values that one would have recorded
by measuring the two observables $Q_\sma$ and $P_\sma$ separately 
in a standard (single) projective measurement on $\varrho_\sma$. 
On the other hand, the two marginal distributions
$$
p_\smx (x) = \int \! dy\, p_\smz (x+ i y) \qquad 
p_\smy (y) = \int \! dx\, p_\smz (x+ i y)\,,
$$
need not to reproduce the distributions obtained in single measurements.
In particular, for the measured variances 
$\langle \Delta X_\smab^2\rangle = 
\langle X_\smab^2\rangle - \langle X_\smab\rangle^2$
and $\langle\Delta Y_\smab\rangle$ one obtains
\begin{align}
\langle \Delta X_\smab^2\rangle &= \Tr \left[
(Q_\sma^2\otimes \id_\smb + \id_\sma \otimes Q_\smb^2 + 2\, Q_\sma \otimes
Q_\smb)\, \varrho_\sma \otimes \varrho_\smb
\right] - \langle Q_\sma \rangle^2 \notag \\
&= \langle \Delta Q_\sma^2 \rangle + \langle Q_\smb^2\rangle \notag 
\\\langle \Delta Y_\smab^2\rangle &= 
\langle \Delta P_\sma^2 \rangle + \langle P_\smb^2\rangle\, 
\label{j8}
\end{align}
where we have already taken into account that $\langle Q_\smb \rangle =
\langle P_\smb\rangle =0$. As it is apparent from Eqs. (\ref{j8}) the 
variances of $X_\smab$ and $Y_\smab$ are larger than those of the 
original, non commuting, observables $Q_\sma$ and $P_\sma$.
\par
Overall, we may summarize the emerging picture as follows: a
joint measurement of a pair of non commuting observables corresponds
to a generalized measurement and may be implemented as the measurement
of a pair of commuting observables on an enlarged Hilbert space.
Mean values are preserved whereas the non commuting nature of the original
observables manifests itself in the broadening of the marginal
distributions, i.e. as an additional noise term appears to both 
the variances. The uncertainty product may be written as
\begin{align}
\langle \Delta X_\smab^2\rangle \langle \Delta Y_\smab^2\rangle & 
= \langle \Delta Q_\sma^2\rangle \langle \Delta P_\sma^2\rangle +
\langle \Delta Q_\sma^2\rangle \langle P_\smb^2\rangle +
\langle Q_\smb^2\rangle \langle \Delta P_\sma^2\rangle +
\langle Q_\smb^2\rangle \langle P_\smb^2\rangle\,, 
\notag \\
& \geq 
\frac14 \big|[Q_\sma,P_\sma]\big|^2 +
\langle \Delta Q_\sma^2\rangle \langle P_\smb^2\rangle +
\langle Q_\smb^2\rangle \langle \Delta P_\sma^2\rangle +
\langle Q_\smb^2\rangle \langle P_\smb^2\rangle\,, 
\label{j9}
\end{align}
where the last three terms are usually referred to as the {\em added
noise} due to the joint measurement. 
If we perform a joint measurement on a minimum 
uncertainty state (MUS, see Appendix \ref{apUR}) 
for a given pair of observables 
(e.g. a coherent state in the joint measurement of a pair of 
conjugated quadratures of the radiation field) and use a MUS also 
for the preparation of the replica system (e.g. the vacuum), then 
Eq. (\ref{j9})
rewrites as 
\begin{align}
\langle \Delta X_\smab^2\rangle \langle \Delta Y_\smab^2\rangle 
= \big|[Q_\sma,P_\sma]\big|^2\,. 
\label{10}
\end{align}
This is four times the minimum attainable uncertainty product 
in the case of a measurement of a single observable (see Appendix
\ref{apUR}).
In terms of rms' $\Delta X = \sqrt{\langle \Delta X^2\rangle}$ we have
a factor $2$, which is usually referred to as the $3$ dB of added noise
in joint measurements.
The experimental realization of joint measurements of non commuting
observables has been carried out for conjugated quadratures of the
radiation field in a wide range of frequencies ranging from 
radiowaves to the optical domain, see e.g. \cite{wal}.  
\subsection{About the so-called Heisenberg
principle} \label{HP}
Let us start by quoting Wikipedia about the Heisenberg principle
\cite{wikiHP}
\begin{quote}
{\em Published by Werner Heisenberg in 1927, the principle implies that it is
impossible to simultaneously both measure the present position while
"determining" the future momentum of an electron or any other particle
with an arbitrary degree of accuracy and certainty. This is not a
statement about researchers' ability to measure one quantity while
determining the other quantity. Rather, it is a statement about the laws
of physics. That is, a system cannot be defined to simultaneously
measure one value while determining the future value of these pairs of
quantities. The principle states that a minimum exists for the product
of the uncertainties in these properties that is equal to or greater
than one half of  the reduced Planck constant.
}\end{quote}
As is it apparent from the above formulation, the principle 
is about the precision achievable in the measurement 
of an observable and the disturbance introduced by the same 
measurement on the state under investigation, which, in turn, 
would limit the precision of a subsequent measurement of the conjugated
observable. The principle, which has been quite useful in the 
historical development of quantum mechanics, has been inferred from 
the analysis of the celebrated Heisenberg' gedanken experiments, and thus 
is heuristic in  nature. However, since its mathematical 
formulation is related to that of the uncertainty relations (see
Appendix \ref{apUR}), it is
often though as a theorem following from the axiomatic structure of 
quantum mechanics.
This is not the case: here we exploit the
formalism of generalized measurements to provide 
an explicit
example of a measurement scheme providing the maximum information
about a given observable, i.e. the statistics of the corresponding PVM,
while leaving the state under investigation in an eigenstate of the 
conjugated observable.
\par
Let us consider the two noncommuting observables $[A,B]= c\, \id$
and the set of detection operators $M_a = |b\rangle\langle a|$ where
$|a\rangle$ and $|b\rangle$ are eigenstates of $A$ and $B$ respectively, 
i.e.  $A|a\rangle=a|a\rangle$, $B|b\rangle=b|b\rangle$. According to
the Naimark theorem the set of operators $\{M_a\}$ describe a
generalized measurement (e.g. an indirect measurement as the one 
depicted in Fig. 1) with statistics $p_a = \Tr [\varrho\,
\Pi_a]$ described by the POVM 
$\Pi_a = M^\dag_a M_a = |a\rangle\langle a|$ and where the conditional 
states after the measurement are given by 
$\varrho_a = \frac{1}{p_a} M_a \varrho M_a^\dag = |b\rangle\langle b|$.
In other words, the generalized measurement described by the set
$\{M_a\}$ has the same statistics of a Von-Neumann projective
measurement of the observable $A$, and leave the system under
investigating in an eigenstate of the observable $B$, thus 
{\em determining its future value with an arbitrary degree of accuracy
and certainty} and contrasting the formulation of the so-called Heisenberg 
principle reported above. 
An explicit unitary realization of this kind of measurement for the
case of position, as well as a detailed discussion on the 
exact meaning of the Heisenberg principle, and the tradeoff 
between precision and disturbance in a quantum measurement, 
may be found in \cite{Ozawa02}. 
\subsection{The quantum roulette}
Let us consider $K$ projective measurements corresponding to $K$
nondegenerate isospectral observables $X_k$, $k=1,...,K$ in a Hilbert
space $H$, and consider the following experiment.  The system is
sent to a detector which at random, with probability $z_k$, $\sum_k
z_k=1$, perform the measurement of the observable $X_k$. This is known
as the quantum roulette since the observable to be measured is chosen
at random, eg according to the outcome of a random generator like a
roulette. The probability of getting the outcome $x$ from the
measurement of the observable $X_k$ on a state $\varrho\in L(H)$ is
given by $p_x^{(k)} = \hbox{Tr}[\varrho\,P^{(k)}_x]$,
$P^{(k)}_x=|x\rangle_k{}_k\langle x|$, and the overall probability
of getting the outcome $x$ from our experiment is given by $$
p_x = \sum_k z_k p_x^{(k)}=\sum_k z_k \hbox{Tr}[\varrho\,P^{(k)}_x] =
\hbox{Tr}[\varrho\,\sum_k z_k P^{(k)}_x] =
\hbox{Tr}[\varrho\,\Pi_x]\,,$$
where the POVM describing the measurement is given by $\Pi_x=\sum_k z_k
P^{(k)}_x$. This is indeed a POVM and not a projective measurement since
$$\Pi_x\Pi_{x^\prime} = \sum_{kk^\prime} z_k z_{k^\prime}
P^{(k)}_xP^{(k^\prime)}_{x^\prime}\neq \delta_{xx^\prime} \Pi_x\,.$$  
Again, we have a practical
situation where POVMs naturally arise in order to describe the
statistics of the measurement in terms of the Born rule and the 
system density operator.
A Naimark extension for the quantum roulette may be obtained as follows.
Let us consider an additional {\em probe} system described by the 
Hilbert space $H_\smp$ of dimension $K$ equal to the number of measured 
observables in the roulette, and the set of projectors $Q_x=\sum_k
P^{(k)}_x \otimes |\theta_k\rangle\langle\theta_k|$ where
$\{|\theta_k\rangle\}$ is a basis for $H_\smp$. Then, upon preparing the 
probe system in the superposition $|\omega_P\rangle=\sum_k \sqrt{z_k}
|\theta_k\rangle$ we have that $p_x=\hbox{Tr}_{\sms\smp} [\varrho\otimes
|\omega_\smp\rangle\langle\omega_\smp|\, Q_x]$ and, in turn,
$\Pi_x = \hbox{Tr}_\smp[\id_\sms\otimes|\omega_\smp\rangle\langle
\omega_\smp|\, Q_x]=\sum_k z_k P^{(k)}_x$. The state of the system after
the measurement may be obtained as the partial trace 
\begin{align}\notag
\varrho_x &= \frac1{p_x} \hbox{Tr}_\smp \left[Q_x\,\varrho\otimes|\omega_\smp
\rangle\langle\omega_\smp|\,Q_x \right] \\\notag & = \frac1{p_x} \sum_k
\sum_{k^\prime} 
\hbox{Tr}_\smp \left[ P_x^{(k)}\otimes|\theta_k\rangle\langle\theta_k| 
\,\varrho\otimes|\omega_\smp \rangle\langle\omega_\smp|\,
P_x^{(k^\prime)}\otimes|\theta_{k^\prime}\rangle\langle\theta_{k^\prime}| 
\right] \\ \notag & =  \frac1{p_x} \sum_k z_k 
P_x^{(k)} \varrho\, P_x^{(k)}\:.
\end{align}
Notice that the presented Naimark extension is not the canonical one.
\begin{exercise}
Prove that the operators $Q_x$ introduced for the Naimark 
extension of the quantum roulette, are indeed
projectors.
\end{exercise}
\begin{exercise} Take a system made by a single qubit system and 
construct the canonical Naimark extension for the quantum roulette 
obtained by measuring the observables $\sigma_\alpha=\cos\alpha\,
\sigma_1 + \sin\alpha\,\sigma_2$, where $\sigma_1$ and $\sigma_2$ 
are Pauli matrices and $\alpha\in[0,\pi]$ is chosen at random 
with probability density $p(\alpha)=\pi^{-1}$.
\end{exercise}
\section{Quantum operations}
In this Section we address the dynamical evolution of quantum
systems to see whether the standard formulation in terms of unitary
evolutions needs a suitable generalization.  This is indeed the case: we
will introduce a generalized description and see how this reconciles
with what we call Postulate \ref{pqd} in the Introduction.  
We will proceed in close
analogy with what we have done for states and measurements. We start by
closely inspecting the physical motivations behind any mathematical
description of quantum evolution, and look for physically motivated
conditions that a map, intended to transform a quantum state into a
quantum state, from now on a {\em quantum operation}, should
satisfy to be admissible. This will lead us to the concept of complete
positivity, which suitably generalizes the motivations behind unitarity. We
then prove that any quantum operation may be seen as the partial trace 
of a unitary evolution in a larger Hilbert space, and illustrate a 
convenient form, the so-called Kraus or operator-sum representation, 
to express the action of a quantum operation on quantum states.
\par
By quantum operation we mean a map $\varrho \rightarrow \cE (\varrho)$
transforming a quantum state $\varrho$ into another quantum state 
$\cE (\varrho)$. The basic requirements on $\cE$ to describe
a physically admissible operations are those captured by the request
of unitarity in the standard formulation, i.e. 
\begin{itemize}
\item[${\boldsymbol{Q1}}$] The map is positive and trace-preserving,
i.e. $\cE (\varrho) \geq 0$ (hence selfadjoint) and $\hbox{Tr} 
[\cE (\varrho)] = \hbox{Tr}[\varrho]=1$. 
The last assumption may be relaxed to that of being trace 
non-increasing $0\leq \hbox{Tr} [\cE
(\varrho)]\leq 1$ in order to include evolution induced by measurements
(see below).
\item[${\boldsymbol{Q2}}$] The map is linear $\cE(\sum_k p_k
\varrho_k) = \sum_k p_k \cE(\varrho_k)$, i.e. the state obtained 
by applying the map to the ensemble $\{p_k, \varrho_k\}$ is 
the ensemble $\{p_k, \cE(\varrho_k)\}$.
\item[${\boldsymbol{Q3}}$] The map is completely positive (CP), 
i.e. besides being positive it is such that if we introduce an additional
system, any map of the form $\cE \otimes \id $ acting on the extended 
Hilbert space is also positive.  In other words, we ask that the map
is physically meaningful also when acting on a portion of a larger, 
composite, system. As we will see, this request is not trivial at all, 
i.e. there exist maps that are positive but not completely positive.
\end{itemize}
\subsection{The operator-sum representation} 
This section is devoted to state and prove a theorem showing that 
a map is a quantum operation if and only if it is the partial trace 
of a unitary evolution in a larger Hilbert space, and provides a 
convenient form, the so-called Kraus decomposition or operator-sum 
representation \cite{Pre,nota}, to express its action on quantum states. 
\begin{theorem}[Kraus] A map $\cE$ 
is a quantum operation i.e. it satisfies the 
requirements $\boldsymbol{Q1}$-$\boldsymbol{Q3}$ {if and only
if} is the partial trace of a unitary evolution on a larger Hilbert
space with factorized initial condition or, equivalently, 
it possesses a Kraus decomposition i. e. 
its action may be represented as $\cE(\varrho) = \sum_k M_k
\varrho M^\dag_k$ where $\{M_k\}$ is a set of operators 
satisfying $\sum_k M_k^\dag M_k=\id$.
\end{theorem}
\begin{proof} The first part of the theorem consists in assuming that 
$\cE(\varrho)$ is the partial trace of a unitary operation in a larger 
Hilbert space and prove that it has a Kraus decomposition and, in turn,
it satisfies the requirements $\boldsymbol{Q1}$-$\boldsymbol{Q3}$.
Let us consider a physical system $A$ prepared in the 
quantum state $\varrho_\sma$ and another system $B$ prepared in the 
state $\varrho_\smb$. $A$ and $B$ interact through the unitary 
operation $U$ and we are interested in describing the effect of this 
interaction on the system $A$ only, i.e. we are looking for the expression 
of the mapping $\varrho_\sma \rightarrow \varrho^\prime_\sma = 
\cE(\varrho_\sma)$ induced by the interaction. This may be obtained by 
performing the partial trace over the system $B$ of the global $AB$ 
system after the interaction, in formula
\begin{align}\notag
\cE(\varrho_\sma) &= \hbox{Tr}_\smb \left[U\, \varrho_\sma \otimes \varrho_\smb U^\dag 
\right] = \sum_{s} p_s  \hbox{Tr}_\smb \left[U\, \varrho_\sma \otimes 
|\theta_s\rangle\langle\theta_s| U^\dag \right] \\ &= \sum_{st} 
p_s \langle\varphi_t | U | \theta_s \rangle \, \varrho_\sma \langle
\theta_s | U^\dag | \varphi_t\rangle = \sum_k M_k
\,\varrho_\sma M^\dag_k
\label{add1}
\end{align}
where we have 
introduced the operator $M_{k}
=\sqrt{p_s}\langle\varphi_t|U|\theta_s\rangle$, with the polyindex $k\equiv st$
obtained by a suitable ordering, and used 
the spectral decomposition of the density operator 
$\varrho_\smb = \sum_s p_s |\theta_s\rangle\langle\theta_s |$.
Actually, we could have also assumed the additional system in a pure 
state $|\omega_\smb\rangle$, since this is always possible upon invoking a 
purification, i.e. by suitably enlarging the Hilbert space. In this case 
the elements in the Kraus decomposition of our map would have be written 
as $\langle\varphi_t| U|\omega_\smb\rangle $.
The set of operators $\{M_k\}$ satisfies the relation $$\sum_k M^\dag
M_k = \sum_{st} p_s \theta_s | U^\dag | \varphi_t\rangle\langle\varphi_t | U 
| \theta_s \rangle = \sum_{s} p_s \langle \theta_s | U^\dag  U | \theta_s
\rangle =\id\,.$$
Notice that the assumption of a factorized initial state 
is crucial to prove the existence of a Kraus decomposition and, 
in turn, the complete positivity. In fact, the dynamical map 
$\cE(\varrho_\sma) =  \hbox{Tr}_\smb \left[U\, \varrho_\smab\,  
U^\dag\right]$ resulting from the partial trace of an 
initially correlated  preparation $\varrho_\smab$ needs not to be so. 
In this case, the dynamics can properly be defined only on a subset 
of initial states of the system. Of course, the map can be extended 
to all possible initial states by linearity, but the 
extension may not be physically realizable, i.e. may be not completely
positive or even positive \cite{PP94}.
\par
We now proceed to show that for map of the form (\ref{add1}) 
(Kraus decomposition) the
properties  $\boldsymbol{Q1}$-$\boldsymbol{Q3}$ hold. Preservation of 
trace and of the Hermitian character, as well as linearity, are
guaranteed by the very form of the map. Positivity is also ensured,
since for any positive operator $O_\sma\in L(H_\sma)$ and any vector 
$|\varphi_\sma\rangle \in H_\sma$  we have
\begin{align}
\langle\varphi_\sma| \cE(O_\sma)|\varphi_\sma\rangle 
&= 
\langle\varphi_\sma|\sum_k M_k\, O_\sma M_k^\dag|\varphi_\sma\rangle 
= 
\langle\varphi_\sma|\hbox{Tr}_\smb[ U\, O_\sma\otimes\varrho_\smb\, U^\dag]|\varphi_\sma\rangle 
\notag \\
&= 
\hbox{Tr}_{\sma\smb}[U^\dag|\varphi_\sma\rangle\langle\varphi_\sma|
\otimes \id\, U\, O_\sma\otimes\varrho_\smb\, ]
\geq 0 \quad\forall\, O_\sma, \forall\, \varrho_\smb, \forall\,
|\varphi_\sma\rangle \,.\notag\end{align} 
Therefore it remains 
to be proved that
the map is completely positive. To this aim let us consider a positive
operator $O_\smac \in L(H_\sma\otimes H_\smc)$ and 
a generic state $|\psi_{\smac}\rangle\rangle$ on the same enlarged space,
and define 
$$|\omega_k \rangle\rangle= \frac1{\sqrt{N_k}}M_k \otimes \id_\smc
|\psi_{\sma\smc}\rangle\rangle \qquad N_k=\langle\langle\psi_{\sma\smc}|
M_k^\dag M_k \otimes \id_\smc|\psi_{\sma\smc}\rangle\rangle\geq0\,.$$
Since $O_\smac$ is positive we have 
$$
\langle\langle\psi_{\smac}|  (M_k^\dag\otimes \id_\smc)
\, O_\smac (M_k \otimes \id_\smc) | \psi_{\smac}\rangle\rangle = 
N_k \langle\langle\omega_k | O_\smac | \omega_k\rangle\rangle \geq 0 
$$ and therefore $\langle\langle\psi_{\smac}|\cE \otimes \id_\smc
(O_\smac)
|\psi_{\smac}\rangle\rangle = \sum_k 
N_k \langle\langle\omega_k | O_\smac | \omega_k\rangle\rangle \geq 0 $, which
proves that for any positive $O_\smac$ also $\cE \otimes \id_\smc
(O_\smac)$ is
positive for any choice of $H_\smc$, i.e. $\cE$ is a CP-map. 
\par
Let us now prove the second part of the theorem, i.e. we consider a 
map $\cE:L(H_\sma) \rightarrow L(H_\sma)$ 
satisfying the requirements $\boldsymbol{Q1}$-$\boldsymbol{Q3}$ and show
that it may be written in the Kraus form and, in turn, that its action 
may be obtained as the partial trace of a unitary evolution in a larger Hilbert.
We start by considering the state $|\varphi\rangle\rangle =
\frac{1}{\sqrt{d}}\sum_k |\theta_k\rangle\otimes |\theta_k\rangle\in
H_\sma\otimes H_\sma$ and define the operator $\varrho_{\sma\sma} = 
\cE \otimes \id (|\varphi\rangle\rangle\langle\langle \varphi|)$. 
From the complete positivity and trace preserving properties of 
$\cE$ we have that $\hbox{Tr}[\varrho_{\sma\sma}]=1$, and 
$\varrho_{\sma\sma}\geq 0$, i.e.
$\varrho_{\sma\sma}$ is a density operator. Besides, this establishes a
one-to-one correspondence between maps $L(H_\sma) \rightarrow L(H_\sma)$
and density operators in $L(H_\sma) \otimes L(H_\sma)$ which may
be proved as follows: for any $|\psi\rangle=\sum_k \psi_k| \theta_k\rangle\in
H_\sma$ define $|\tilde\psi\rangle=\sum_k \psi_k^* | \theta_k\rangle$
and notice that 
$$
\langle\tilde\psi| \varrho_{\sma\sma} | \tilde\psi\rangle = \frac1d  
\langle\tilde\psi| \sum_{kl}\cE(|\theta_k \rangle\langle\theta_l|) \otimes 
|\theta_k\rangle\langle\theta_l| \,|\tilde\psi\rangle 
=\frac1d \sum_{kl} \psi_l^* \psi_k\, \cE(|\theta_k\rangle\langle\theta_l|) =
\frac1d\,\cE(|\psi\rangle\langle\psi |)\,,
$$
where we used linearity to obtain the last equality. Then define 
the operators $M_k|\psi\rangle = \sqrt{d p_k}\langle\tilde\psi |
\omega_k\rangle\rangle$, where $|\omega_k\rangle\rangle$ are the
eigenvectors of $\varrho_{\sma\sma}=\sum_k
p_k|\omega_k\rangle\rangle\langle\langle\omega_k| $: this is a linear 
operator on $H_\sma$ and we have
$$
\sum_k M_k |\psi\rangle\langle\psi| M_k^\dag = d \sum_k p_k
\langle\tilde\psi|\omega_k\rangle\rangle\langle\langle\omega_k|\tilde\psi\rangle
= d \langle\tilde\psi|\varrho_{\sma\sma}|\tilde\psi\rangle =
\cE(|\psi\rangle\langle\psi |)
$$
for all pure states. Using again linearity we have that $\cE(\varrho) =
\sum_k M_k \varrho M^\dag_k$ also for any mixed state. It remains to be
proved that a unitary extension exists, i.e. to
prove that for any map on $L(H_\sma)$ which satisfies
$\boldsymbol{Q1}$-$\boldsymbol{Q3}$, and thus possesses a Kraus
decomposition, there exist: i) a Hilbert 
space $H_\smb$, ii) a state $|\omega_\smb\rangle \in H_\smb$, iii) a unitary 
$U\in L(H_\sma\otimes H_\smb)$ such that $\cE(\varrho_\sma)
=\hbox{Tr}_\smb [U\,\varrho_\sma \otimes
|\omega_\smb\rangle\langle\omega_\smb| U^\dag]$
for any $\varrho_\sma \in L(H_\sma)$. To this aim we proceed as we did
for the proof of the Naimark theorem, i.e. we take an arbitrary state 
$|\omega_\smb\rangle \in H_\smb$, and define an operator $U$ trough its
action on the generic $\varphi_\sma\rangle \otimes |\omega_\smb
\rangle\in H_\sma \otimes H_\smb$, 
$
U\,|\varphi_\sma\rangle \otimes |\omega_\smb \rangle = \sum_k
M_k\,|\varphi_\sma\rangle \otimes |\theta_k\rangle$, 
where the $|\theta_k\rangle$'s are a basis for $H_\smb$.
The operator $U$ preserves the scalar product
\begin{align}\notag
\langle\langle \omega_\smb,\varphi_\sma^\prime | U^\dag U|
\varphi_\sma,\omega_\smb \rangle\rangle 
= \sum_{k k ^\prime}
\langle \varphi_\sma^\prime | M_{k^\prime}^\dag
M_k|\varphi_\sma\rangle \langle \theta_{k^\prime}|\theta_k\rangle
= \sum_{k}
\langle\varphi_\sma^\prime | M_{k}^\dag
M_k|\varphi_\sma\rangle
= 
\langle\varphi_\sma^\prime | \varphi_\sma\rangle
\end{align}
and so it is unitary in the one-dimensional subspace spanned by 
$|\omega_\smb\rangle$. Besides, it may be extended to a full 
unitary operator in the global Hilbert space $H_\sma\otimes H_\smb$, 
eg it can be the identity operator in the subspace orthogonal to
$|\omega_\smb\rangle$. Then, for any $\varrho_\sma$ in $H_\sma$ we have
\begin{align}\notag
\hbox{Tr}_\smb \left[
U \varrho_\sma \otimes |\omega_\smb\rangle\langle\omega_\smb |\, U^\dag
\right]&= 
\sum_s p_s\, \hbox{Tr}_\smb \left[
U |\psi_s\rangle\langle\psi_s | \otimes
|\omega_\smb\rangle\langle\omega_\smb |\, U^\dag\right]
\\ \notag &=
\sum_{skk^\prime} p_s\, \hbox{Tr}_\smb \left[
M_k |\psi_s\rangle\langle\psi_s |\,M_{k^\prime}^\dag \otimes
|\theta_k\rangle\langle\theta_{k^\prime}| \right]
\\ \notag &=
 \sum_{sk} p_s\, M_k |\psi_s\rangle\langle\psi_s |\,M_{k}^\dag
= \sum_{k}  M_k \varrho_\sma M_{k}^\dag
\qquad \qed
\end{align}
\end{proof}
The Kraus decomposition of a quantum operation generalizes the 
unitary description of quantum evolution. Unitary maps are, of course,  
included and correspond to maps whose Kraus decomposition contains a 
single elements. The set of quantum operations constitutes a semigroup, 
i.e. the composition 
of two quantum operations is still a quantum 
operation: $$\cE_2 (\cE_1(\varrho))
= \sum_{k_1} M^{(1)}_{k_1} \cE_2(\varrho) M^{(1)\dag}_{k_1}= 
\sum_{k_1k_2} 
M^{(1)}_{k_1}
M^{(2)}_{k_2}
\varrho
M^{(2)\dag}_{k_2}
M^{(1)\dag}_{k_1}
=\sum_{\boldsymbol{k}} 
\boldsymbol{M}_{\boldsymbol{k}} 
\varrho
\boldsymbol{M}_{\boldsymbol{k}}^\dag\,, $$ 
where we have introduced the polyindex $\boldsymbol{k}$. Normalization
is easily proved, since 
$\sum_{\boldsymbol{k}} \boldsymbol{M}_{\boldsymbol{k}}^\dag 
\boldsymbol{M}_{\boldsymbol{k}} = \sum_{k_1k_2} M^{(2)\dag}_{k_2}
M^{(1)\dag}_{k_1} M^{(1)}_{k_1} M^{(2)}_{k_2}=\id$.
On the other hand, the existence of 
inverse is not guaranteed:
actually only unitary operations are invertible (with a CP inverse).
\par
The Kraus theorem also allows us to have a unified picture of quantum
evolution, either due to an interaction or to a measurement. In fact,  
the modification of the state in the both processes is described by a set of 
operators $M_k$ satisfying $\sum_k M^\dag_k M_k = \id$. In this
framework, the Kraus operators of a measurement are what we have
referred to as the detection operators of a POVM.
\subsubsection{The dual map and the unitary equivalence}
Upon writing the generic expectation value for the evolved 
state $\cE(\varrho)$ and exploiting both linearity and circularity 
of trace we have 
$$\langle X\rangle = \hbox{Tr} [\cE(\varrho )\, X]=
\sum_k \hbox{Tr}[M_k \varrho M_k^\dag \, X]
= \sum_k \hbox{Tr}[\varrho\, M_k^\dag X M_k]
=\hbox{Tr}[\varrho \cE^\vee (X)]\,,$$ where we have defined 
the dual map $\cE^\vee (X)=\sum_k M_k^\dag X M_k$ which represents
the "Heisenberg picture" for quantum operations.
Notice also that the elements of the Kraus decomposition 
$M_k=\langle\varphi_k | U|\omega_\smb\rangle$ depend on
the choice of the basis used to perform the partial trace.
Change of basis cannot have a physical effect and this means 
that the set of operators $$N_k=\langle\theta_k| U|\omega_\smb\rangle = 
\sum_s \langle\theta_k|\varphi_s\rangle\langle\varphi_s| U|\omega_\smb\rangle = 
\sum_s V_{ks} M_s \,,$$ where the unitary $V\in L(H_\smb)$ 
describes the change of basis, and the original set $M_k$ 
actually describe the same 
quantum operations, i.e. 
$\sum_k N_k \varrho N_k^\dag=\sum_k M_k \varrho M_k^\dag$, $\forall
\varrho$. 
The same can be easily proved for the system $B$ prepared in mixed state.
The origin of this degree of freedom stays in the fact that if the 
unitary $U$ on $H_\sma \otimes H_\smb$ and the state
$|\omega_\smb\rangle\in H_\smb$ realize an extension for the map
$\cE:L(H_\sma) \rightarrow L(H_\sma)$ then any unitary of the form $(\id
\otimes V) U$ is a unitary extension too, with the same ancilla state.
A quantum operation is thus identified by an equivalence class of 
Kraus decompositions. 
An interesting corollary is that any quantum operation on a given
Hilbert space of dimension $d$ may be generated by a Kraus decomposition
containing at most $d^2$ elements, i.e. given a Kraus decomposition 
$\cE(\varrho) = \sum_k M_k \varrho M_k^\dag$ with an arbitrary number of
elements, one may exploit the unitary equivalence and find another 
representation $\cE(\varrho) = 
\sum_k N_k \varrho N_k^\dag$ with at most $d^2$ elements.
\subsection{The random unitary map and the depolarizing channel}
A simple example of quantum operation is the random unitary map,
defined by the Kraus decomposition $\cE(\varrho) = \sum_k p_k U_k \varrho
U^\dag_k$, i.e. $M_k=\sqrt{p_k}\, U_k$ and $U_k^\dag U_k=\id$. This map 
may be seen as the evolution resulting from the interaction of our system
with another system of dimension equal to the number of elements in the 
Kraus decomposition of the map via the unitary $V$ defined by
$V|\psi_\sma\rangle\otimes|\omega_\smb\rangle=\sum_k \sqrt{p_k}\, U_k
|\psi_\sma\rangle\otimes|\theta_k\rangle$, $|\theta_k\rangle$ being a
basis for $H_\smb$ which includes $|\omega_\smb\rangle$. 
If "we do not look" at the system $B$ and trace out
its degree of freedom the evolution of system $A$ is governed by
the random unitary map introduced above. 
\begin{exercise} Prove explicitly the unitarity of V. \end{exercise} 
The operator-sum representation of quantum evolutions have been
introduced, and finds its natural application, for the description 
of propagation in noisy channels, i.e. the evolution resulting from 
the interaction of the system of interest with an external environment, 
which generally introduces noise in the system degrading its
coherence.
As for example, let us consider a qubit system (say, the
polarization of a photon), on which we have encoded binary information 
according to a suitable coding procedure, traveling from a sender to a 
receiver. The propagation needs a physical support (say, an optical
fiber) and this unavoidably leads to consider possible perturbations
to our qubit, due to the interaction with the environment. The resulting 
open system dynamics is usually governed by a Master equation, i.e. the
equation obtained by partially tracing the Schroedinger 
(Von Neumann) equation governing the dynamics of  the global system, 
and the solution 
is expressed in form 
of a CP-map. For a qubit $Q$ in a noisy environment a quite general
description of the detrimental effects of the environment is the
so-called depolarizing channel \cite{nie00}, 
which is described by the Kraus operator
$M_0 = \sqrt{1-\gamma}\,\sigma_0$, $M_k=\sqrt{\gamma/3}\,\sigma_k$,
$k=1,2,3$, i.e. $$\cE(\varrho) = (1-\gamma) \varrho + \frac{\gamma}{3}
\sum_k \sigma_k\, \varrho\,\sigma_k \qquad 0\leq \gamma\leq 1\,.$$ 
The depolarizing channel may be seen
as the evolution of the qubit due to the interaction with a
four-dimensional system through the unitary
$$V|\psi_\smq\rangle\otimes|\omega_\sme\rangle =\sqrt{1-\gamma}
|\psi_\smq\rangle\otimes|\omega_\sme\rangle  +  \sqrt{\frac{\gamma}3} 
\sum_{k=1}^3 \sigma_k 
|\psi_\smq\rangle\otimes|\theta_k\rangle\,,$$ $|\theta_k\rangle$ being a
basis which includes $|\omega_\sme\rangle$. 
From the practical point view, the map describes a situation in 
which, independently on the underlying physical mechanism, we have a
probability $\gamma/3$ that a perturbation described by a Pauli matrix
is applied to the qubit. If we apply $\sigma_1$  we have the so-called 
spin-flip i.e. the exchange $|0\rangle \leftrightarrow |1\rangle$,
whereas if we apply $\sigma_3$ we have the phase-flip, and for $\sigma_2$ 
we have a specific combination of the two effects.
Since for any state of a qubit $\varrho + \sum_k
\sigma_k\varrho\sigma_k= 2\id$ the action of the depolarizing channel
may be written as 
$$\cE(\varrho) = (1-\gamma) \varrho + \frac{\gamma}{3}
(2 \id -\varrho) = \frac23 \gamma \id + (1- \frac43 \gamma )\varrho = 
p\varrho + (1-p) \frac{\id}2\,,$$ 
where $p=1-\frac43 \gamma$, i.e. $-\frac13\leq
p\leq 1$. In other words, we have that the original state $\varrho$ is sent
to a linear combination of itself and the maximally mixed state
$\frac{\id}2$, also referred to as the depolarized state.
\begin{exercise}
Express the generic qubit state in Bloch
representation and explicitly write the effect of the depolarizing
channel on the Bloch vector. 
\end{exercise}
\begin{exercise} Show that the purity of a qubit cannot increase
under the action of the depolarizing channel.
\end{exercise}
\subsection{Transposition and partial transposition}
The transpose $T(X)=X^\smt$ of an operator $X$ is the conjugate of its 
adjoint $X^\smt = (X^\dag)^* = (X^*)^\dag$. Upon the choice of a basis 
we have
$X=\sum_{nk} X_{nk} |\theta_n\rangle\langle\theta_k |$ and thus
$X^\smt=\sum_{nk} X_{nk} |\theta_k\rangle\langle\theta_n |
=\sum_{nk} X_{kn} |\theta_n\rangle\langle\theta_k |$.
Transposition does not change the trace of an operator, neither 
its eigenvalues. Thus it transforms density operators into density 
operators: $\hbox{Tr}[\varrho]=\hbox{Tr}[\varrho^\smt]=1$
$\varrho^\smt \geq 0$ if $\varrho\geq 0$. As a positive, trace
preserving, map it is a candidate  to be a quantum operation. 
On the other hand, we will show by a counterexample that it fails to
be completely positive and thus it does not correspond to physically
admissible quantum operation. 
Let us consider a bipartite system formed by two qubits prepared in 
the state $|\varphi\rangle\rangle=\frac1{\sqrt{2}}\, 
|00\rangle\rangle +|11\rangle\rangle$. We denote by $\varrho^\tau = \id
\otimes T (\varrho)$ the partial transpose of $\varrho$ i.e. the operator
obtained by the application of the transposition map to one of the two
qubits. We have 
\begin{align} \notag
\big(|\varphi\rangle\rangle\langle
\langle\varphi|\big)^\tau
&= \frac12 
\left(
\begin{array}{cccc}
1 & 0 & 0 & 1 \\ 
0 & 0 & 0 & 0 \\ 
0 & 0 & 0 & 0 \\ 
1 & 0 & 0 & 1 
\end{array}
\right)^\tau
\\ \notag  
& =\frac12 
\Big(
|0\rangle\langle 0| \otimes|0\rangle\langle 0| + 
|1\rangle\langle 1| \otimes|1\rangle\langle 1| +
|0\rangle\langle 1| \otimes|0\rangle\langle 1| + 
|1\rangle\langle 0| \otimes|1\rangle\langle 0|  
\Big)^\tau
\\ \notag  
& =\frac12 
\Big(
|0\rangle\langle 0| \otimes|0\rangle\langle 0| + 
|1\rangle\langle 1| \otimes|1\rangle\langle 1| +
|0\rangle\langle 1| \otimes|1\rangle\langle 0| + 
|1\rangle\langle 0| \otimes|0\rangle\langle 1|  
\Big) 
\\ \notag  
& =\frac12 
\left(
\begin{array}{cccc}
1 & 0 & 0 & 0 \\ 
0 & 0 & 1 & 0 \\ 
0 & 1 & 0 & 0 \\ 
0 & 0 & 0 & 1 
\end{array}
\right)
\end{align}
Using the last expression it is straightforward to evaluate the eigenvalues 
of $\varrho^\tau$, which are $+\frac12$ (multiplicity three) and
$-\frac12$. In other words $\id \otimes T$ is not a positive map and the
transposition is not completely positive. 
Notice that for a factorized state of the form
$\varrho_\smab=\varrho_\sma \otimes \varrho_\smb$ we have 
$\id \otimes T (\varrho_\smab) = \varrho_\sma \otimes \varrho_\smb^\smt
\geq 0$ i.e. partial transposition preserves positivity in this case .
\begin{exercise}
Prove that transposition is not a CP-map 
by its action on any state of the form $
|\varphi\rangle\rangle = 
\frac1{\sqrt{d}} \sum_k |\varphi_k\rangle\otimes 
|\theta_k\rangle$. Hint: the operator 
$\id \otimes T (|\varphi\rangle\rangle\langle\langle\varphi|)\equiv E$
is the so-called swap operator since it "exchanges" states
as $E(|\psi\rangle_\sma\otimes|\varphi\rangle_\smb) = |\varphi\rangle_\sma
\otimes |\psi\rangle_\smb$.
\end{exercise}
\section{Conclusions}
In this tutorial, we have addressed the postulates of quantum 
mechanics about states, measurements and operations. 
We have reviewed their modern formulation and 
introduced the basic mathematical tools: density
operators, POVMs, detection operators and CP-maps. We have 
shown how they
provide a suitable framework to describe quantum systems in interaction
with their environment, and with any kind of measuring and processing
devices. The connection with the standard formulation have been 
investigated in details building upon the concept of purification and 
the Theorems of  Naimark and Stinespring/Kraus-Choi-Sudarshan.
\par
The framework and the tools illustrated in this tutorial are suitable
for the purposes of quantum information science and technology, a field
which has fostered new experiments and novel views on the conceptual
foundation of quantum mechanics, but has so far little impact on the 
way that it is taught. We hope to contribute in disseminating these notions
to a larger audience, in the belief that they are useful for several 
other fields, from condensed matter physics to quantum biology.
\begin{acknowledgement} 
I'm grateful to Konrad Banaszek, Alberto Barchielli, 
Maria Bondani, Mauro D'Ariano, Ivo P.
Degiovanni, Marco Genoni, Marco Genovese, Paolo Giorda, Chiara
Macchiavello, Sabrina Maniscalco, Alex Monras, Stefano Olivares, Jyrki
Piilo, Alberto Porzio, Massimiliano Sacchi, Ole Steuernagel, and Bassano
Vacchini for the interesting and fruitful discussions about foundations
of quantum mechanics and quantum optics over the years. I would also
like to thank Gerardo Adesso, Alessandra Andreoni, Rodolfo Bonifacio, 
Ilario Boscolo, Vlado Buzek, Berge Englert, Zdenek Hradil, Fabrizio 
Illuminati, Ludovico Lanz, 
Luigi Lugiato, Paolo Mataloni, Mauro Paternostro, Mladen Pavi\v{c}i\'{c}, 
Francesco 
Ragusa, Mario Rasetti, Mike Raymer, Jarda \v{R}eh\'{a}\v{c}ek, Salvatore
Solimeno, and Paolo Tombesi.
\end{acknowledgement} 

\section*{Further readings}
{\small
\begin{enumerate}
\item 
I. Bengtsson, K. Zyczkowski, {\em Geometry of Quantum
States}, (Cambridge University Press, 2006).
\item Lectures and reports by C. M. Caves, available at
{\tt http://info.phys.unm.edu/$\,\,\widetilde{}$caves/}
\item
P. Busch, M. Grabowski, P. J. Lahti,{\em Operational Quantum Mechanics}, Lect. 
Notes. Phys. {\bf 31}, (Springer, Berlin,1995). 
\item
T. Heinosaari, M. Ziman, Acta Phys. Slovaca {\bf 58}, 487 (2008). 
\item C. W. Helstrom, {\em Quantum Detection and Estimation 
Theory} (Academic Press, New York, 1976)
\item
A.S. Holevo, {\em Statistical Structure of 
Quantum Theory}, Lect. Not. Phys {\bf 61}, (Springer, Berlin,
2001).
\item M. Ozawa, J.  Math. Phys. \textbf{25}, 79 (1984).
\item 
M. G. A. Paris, J. Rehacek (Eds.), {\em Quantum State Estimation} 
Lect. Notes Phys. {\bf 649}, (Springer, Berlin, 2004).
\item 
V. Gorini, A. Frigerio, M. Verri, A. Kossakowski, E. C. G. 
Sudarshan, Rep. Math. Phys. {\bf 13}, 149 (1978).
\item
F. Buscemi, G. M. D'Ariano, and M. F. Sacchi, Phys. Rev. A {\bf 68}.
042113 (2003).
\item
K. Banaszek, Phys. Rev. Lett. {\bf 86}, 1366 (2001). 
\end{enumerate}
}
\appendix
\section{Trace and partial trace}
\label{apTR}
The trace of an operator $O$ is a scalar quantity equal to sum of diagonal
elements in a given basis $\hbox{Tr}[O]=\sum_n
\langle\varphi_n|O|\varphi_n\rangle$. The trace is invariant under any
change of basis, as it is proved by the following chain of equa\-lities
\begin{align}\sum_n\langle\theta_n|O|\theta_n\rangle &=
\sum_{njk}\langle\theta_n|\varphi_k\rangle
\langle\varphi_k|O|\varphi_j\rangle
\langle\varphi_j|\theta_n\rangle=
\sum_{njk}
\langle\varphi_j|\theta_n\rangle
\langle\theta_n|\varphi_k\rangle
\langle\varphi_k|O|\varphi_j\rangle
\notag \\ 
& \notag =\sum_{jk}
\langle\varphi_j|\varphi_k\rangle
\langle\varphi_k|O|\varphi_j\rangle
=\sum_{k}
\langle\varphi_k|O|\varphi_k\rangle\,,\end{align} 
where we have suitably inserted and removed
resolutions of the identity in terms of both basis $\{|\theta_n\rangle\}$ 
and $\{|\varphi_n\rangle\}$. As a consequence, using the basis of
eigenvectors of $O$, $\hbox{Tr}[O]=\sum_n o_n$, $o_n$ being the
eigenvalues of $O$. Trace is a linear operation, i.e. 
$\hbox{Tr}[O_1+O_2]=\hbox{Tr}[O_1]+\hbox{Tr}[O_2]$ and
$\hbox{Tr}[\lambda\,O]=\lambda\hbox{Tr}[O]$ and thus
$\partial\hbox{Tr}[O]=\hbox{Tr}[\partial O]$ for any derivation.
The trace of any "ket-bra" $\hbox{Tr}[|\psi_1\rangle\langle\psi_2|]$ 
is obtained by "closing the sandwich" $\hbox{Tr}[|\psi_1
\rangle\langle\psi_2|]=\langle\psi_2|\psi_1\rangle$; in fact upon
expanding the two vectors in the same basis and taking the trace
in that basis 
$\hbox{Tr}[|\psi_1\rangle\langle\psi_2|]= \sum_{nkl}
\psi_{1k}\psi_{2l}^*
\langle\theta_n|\theta_k\rangle\langle\theta_l|\theta_n\rangle 
=\sum_{n}\psi_{1n}\psi_{2n}^*
=\langle\psi_2|\psi_1\rangle$.
Other properties are summarized by the 
following theorem.
\begin{theorem}
For the trace operation the following properties hold
\begin{itemize}
\item[i)] Given any pair of operators $\Tr[A_1A_2] =\Tr[A_2A_1]$
\item[ii)] Given any set of operators $A_1,...,A_\smn$ we
$\Tr[A_1A_2A_3...A_\smn] 
= \Tr[A_2A_3...A_\smn A_1]
= \Tr[A_3A_4...A_1A_2]=...$ (circularity). 
\end{itemize}
\end{theorem}
\begin{proof}: left as an exercise.\qed\end{proof}
Notice that the "circularity" condition is essential to have 
property ii) i.e. 
$\hbox{Tr}[A_1A_2A_3] = \hbox{Tr}[A_2A_3A_1]$, but
$\hbox{Tr}[A_1A_2A_3] \neq \hbox{Tr}[A_2A_1A_3]$
\par
Partial traces $R_\smb\in L(H_\smb)$ $R_\sma\in L(H_\sma)$ 
of an operator $R$ in $L(H_1\otimes H_2)$ 
are defined accordingly as 
$$
R_\smb=\hbox{Tr}_\sma \left[R\,\right] = \sum_n
{}_\sma\langle\varphi_n|R\,|\varphi_n\rangle_\sma
\qquad
R_\sma=\hbox{Tr}_\smb \left[R\,\right] = \sum_n
{}_\smb\langle\varphi_n|R\,|\varphi_n\rangle_\smb\,$$
and circularity holds only for single-system operators, e.g., 
if $R_1, R_2 \in L (H_\sma\otimes H_\smb)$, $A\in L(H_\sma)$, $B\in
L(H_\smb)$
\begin{align}\notag
\hbox{Tr}_\sma \left[A\otimes \id \,R_1 R_2\right] &= \sum_n a_n \langle
a_n | R_1 R_2|a_n\rangle = \hbox{Tr}_\sma \left[R_1 R_2\,A\otimes \id \right] 
\\ \notag
\hbox{Tr}_\sma \left[A\otimes B\, R_1 R_2\right] &= 
\sum_n a_n \langle a_n | \id \otimes B\, R_1 R_2|a_n\rangle 
=
\hbox{Tr}_\sma \left[\id \otimes B\, R_1 R_2\, A\otimes \id\right]
\\ \notag
&\neq
\sum_n a_n \langle a_n | R_1 R_2\,\id\otimes B |a_n\rangle 
= \hbox{Tr}_\sma \left[R_1 R_2\,A\otimes B \right] 
\end{align}
\begin{exercise} Consider a generic mixed state $\varrho\in
L(H\otimes H)$ and write the matrix elements of the two partial traces
in terms of the matrix elements of $\varrho$.
\end{exercise}
\begin{exercise} Prove that also partial trace is invariant under
change of basis.
\end{exercise}
\section{Uncertainty relations}
\label{apUR}
Two non commuting observables $[X,Y]\neq 0$ do not admit a 
complete set of common eigenvectors, and thus it not 
possible to find common eigenprojectors
and to define a joint observable. Two non commuting observables are said to
be incompatible or complementary, since they cannot assume definite values
simultaneously. A striking consequence of this fact is that when 
we measure an observable $X$ the precision of the measurement, as 
quantified by the variance  $\langle \Delta X^2\rangle = 
\langle X^2\rangle - \langle X\rangle^2$, is influenced by the 
variance of any observable which is non commuting with
$X$ and cannot be made arbitrarily small. In order to determine the
relationship between the variances of two noncommuting 
observables, one of which is measured
on a given state $|\psi\rangle$, let us consider the two
vectors
$$
|\psi_1\rangle =(X-\langle X\rangle) |\psi\rangle \qquad
|\psi_2\rangle =(Y-\langle Y\rangle) |\psi\rangle\,,
$$
and write explicitly the Schwartz inequality
$\langle\psi_1|\psi_1\rangle\langle\psi_2|\psi_2\rangle \geq
\left|\langle\psi_1|\psi_2\rangle\right|^2$, i.e. \cite{Puri}
\begin{align}
\langle \Delta X^2\rangle 
\langle \Delta Y^2\rangle 
\geq \frac14 \left[ \left|\langle F\rangle \right|^2 + \left|\langle C
\rangle\right|^2\right] \geq \frac14  \left|\langle C
\rangle\right|^2\,, 
\label{UR}
\end{align}
where 
$[X,Y]=iC$ and $F=XY-YX-2\langle X\rangle \langle Y\rangle$.
Ineq. (\ref{UR}) represents the uncertainty relation for the non
commuting observables $X$ and $Y$ and it is 
usually presented in the form involving the second inequality.
Uncertainty relations set a lower bound to the measured 
variance in the measurement of a single observable, say $X$, 
on a state with a fixed, intrinsic, variance of the complementary
observable $Y$ (see Section \ref{jm} for the relationship between the
variance of two non commuting observables in a joint measurement). 
The uncertainty product is minimum when the two vectors
$|\psi_1\rangle$ and $|\psi_2\rangle$ are parallel in the Hilbert space, 
i.e. $|\psi_1\rangle = -i \lambda |\psi_2\rangle$ where $\lambda$ is a 
complex number. Minimum uncertainty states (MUS) for the pair of observables 
$X,Y$ are thus the states satisfying 
$$
\left(X + i \lambda Y \right)|\psi\rangle = \left( \langle X\rangle  + i \lambda
\langle Y \rangle \right) |\psi\rangle\,.
$$
If $\lambda$ is real then $\langle F\rangle=0$, i.e. the quantities $X$
and $Y$ are uncorrelated when the physical system is prepared in the 
state $|\psi\rangle$.  If $|\lambda|=1$ then 
$\langle \Delta X^2\rangle=\langle \Delta Y^2\rangle$ and the 
corresponding states are
referred to as equal variance MUS. Coherent states of a single-mode 
radiation field \cite{cah69} are equal variance MUS, e. g. for the pair of quadrature
operators defined by 
$Q=\frac{1}{\sqrt{2}} (a^\dag + a)$ and 
$P=\frac{i}{\sqrt{2}} (a^\dag - a)$. 

\begin{thebibliography}{99}
\bibitem{nie00} M. Nielsen, E. Chuang, {\em Quantum Computation and
Quantum Information}, (Cambridge University Press, 2000).
\bibitem{Per93} A. Peres, {\em Quantum Theory: concepts and methods}, 
(Kluwer Academic, Dordrecht, 1993).
\bibitem{Bergou} J. Bergou, J. Mod. Opt. {\bf 57}, 160 (2010).
\bibitem{Pau} V. Paulsen, {\em Completely Bounded Maps and Operator
Algebras} (Cambridge University Press, 2003).
\bibitem{art1} E. Arthurs, J. L. Kelly, Bell. Syst. Tech. J. {\bf 44}, 725
(1965); J. P. Gordon, W. H. Louisell in {\em Physics of Quantum 
Electronics} (Mc-Graw-Hill, NY, 1966); 
E. Arthurs, M. S. Goodman, Phys. Rev. Lett. {\bf 60}, 2447 (1988).
\bibitem{yue82} H. P. Yuen, Phys. Lett. A {\bf 91}, 101 (1982).
\bibitem{BV10} B. Vacchini in {\em Theoretical foundations of quantum
information processing and communication}, E. Bruening et al (Eds.),
Lect. Not. Phys. {\bf 787}, 39 (2010).
\bibitem{jsp1} E. Prugove\v{c}ki, J. Phys. A {\bf 10}, 543 (1977).
\bibitem{wal} N. G. Walker, J. E. Carrol, Opt. Quantum Electr.
{\bf 18}, 355 (1986); N. G. Walker, J. Mod. Opt. {\bf 34}, 16 (1987).
\bibitem{wikiHP} \begin{verbatim}
http://en.wikipedia.org/wiki/Uncertainty_principle \end{verbatim}
\bibitem{Ozawa02} M. Ozawa, Phys. Lett. A {\bf 299}, 17 (2002);
Phys. Rev. A {\bf 67}, 042105  (2003); J. Opt. B {\bf 7}, S672 (2005).
\bibitem{Pre} J. Preskill, {\em Lectures notes for Physics 229: Quantum
information and computation} available
at {\tt www.theory.caltech.edu/$\,\,\widetilde{}$preskill/ph229/}
\bibitem{nota}
{Depending
on the source, and on the context, the theorem is known as the 
Stinespring dilation theorem, or the Kraus-Choi-Sudarshan theorem.} 
\bibitem{PP94} P. Pechukas, Phys. Rev. Lett. {\bf 73}, 1060 (1994).
\bibitem{Puri} R. Puri, {\em Mathematical methods of quantum optics}
(Springer, Berlin, 2001).
\bibitem{cah69} K. E. Cahill, R. J. Glauber, Phys. Rev. {\bf 177}, 1857
(1969); {\bf 177}, 1882 (1969).
\end{thebibliography}
\end{document}